\documentclass[11pt, leqno, letterpaper]{amsart}
 
\usepackage{graphicx}    
\usepackage{xcolor}

\usepackage[margin=1.2in]{geometry}

\usepackage{amsmath,amsthm,amsfonts,amssymb}
 
 \usepackage{soul}

\newtheorem{theorem}{Theorem}

\newtheorem{lemma}[theorem]{Lemma}

\newtheorem{proposition}[theorem]{Proposition}

\newtheorem{remark}[theorem]{Remark}

\newtheorem{assumption}[theorem]{Assumption}

\begin{document}

\title[Asian Options with Jumps]
{Asymptotics for Short Maturity Asian Options in Jump-Diffusion models with Local Volatility}

\author{Dan Pirjol}
\address
{School of Business\newline
\indent Stevens Institute of Technology\newline
\indent Hoboken, NJ-07030\newline
\indent United States of America}
\email{dpirjol@gmail.com}

\author{Lingjiong Zhu}
\address
{Department of Mathematics\newline
\indent Florida State University\newline
\indent 1017 Academic Way\newline
\indent Tallahassee, FL-32306\newline
\indent United States of America}
\email{zhu@math.fsu.edu}

\date{27 February 2024}

\subjclass[2010]{91G20,91G80,60J75}
\keywords{Asian options, short maturity, L\'{e}vy jumps, local volatility.}

\begin{abstract} 
We present a study of the short maturity asymptotics for Asian options in a 
jump-diffusion model with
a local volatility component, where the jumps are modeled as a compound Poisson
process. The analysis for out-of-the-money Asian options is extended to models with L\'evy jumps, 
including the exponential L\'{e}vy model as a special case. 
Both fixed and floating strike Asian options are considered.
Explicit results are obtained for the first-order
asymptotics of the Asian options prices for a few popular models in the literature: the Merton jump-diffusion model,
the double-exponential jump model, and the Variance Gamma model. 
We propose an analytical approximation for Asian option prices which satisfies the constraints from the short-maturity asymptotics, and test it against Monte Carlo simulations. The asymptotic
results are in good agreement with numerical simulations 
for sufficiently small maturity.
\end{abstract}

\maketitle

\section{Introduction}

Asian options are popular instruments traded in many financial markets,
on underlyings such as commodity futures, equities, indices and currency
exchange rates (FX). Compared to vanilla European options they have the
advantage that they are less sensitive to short term price fluctuations
of the underlying asset, due to their averaging property. Typically, such
options have a payoff of the form:
\begin{equation}
\mathrm{Payoff} = \left( \theta(A_T - K), 0 \right)^{+}\,,
\end{equation}
where $A_T$ is the average of the asset price $S_t$ over an averaging period
$[0, T]$, $K>0$ is the strike price, $\theta =\pm 1$ for a call/put Asian option, and $x^{+}$ denotes $\max(x,0)$ for any $x\in\mathbb{R}$. 
Although in practice the averaging is in discrete time (daily averaging), 
it is convenient to approximate $A_T$  with the continuous time average:
\begin{equation}\label{A:T:defn}
A_T := \frac{1}{T} \int_{0}^{T} S_t dt \,.
\end{equation}

Most of the theoretical work on pricing Asian options assumes that the asset 
price follows a diffusion process, with continuous sample paths. We summarize briefly a few results using methods which can be extended to models with jumps. 

Geman and Yor \cite{GY} found an exact result for the price of an Asian 
option in the Black-Scholes model with random maturity $T_\lambda \sim\mathrm{Exp}(\lambda)$, distributed according to an exponential distribution with 
parameter $\lambda$. 
This reduces the pricing of such options to the inversion of a Laplace transform. 
See \cite{CS,DufresneReview,FMW} for details of application of this method.

A generalization of this approach was proposed by Cai and Kou \cite{CaiKou}
who showed that the double Laplace transform of the Asian option prices with
respect to strike and maturity has a very simple form and can be computed
exactly in a wide class of models. This reduces the Asian options pricing
problem to the inversion of a double Laplace transform, for which efficient
numerical methods are available. This method was applied in \cite{Hackmann2014}
to models driven by hyper-exponential L\'evy processes, which are generalizations of models with hyper-exponential jumps.

The Cai and Kou approach has been simplified by Cai et al.
\cite{CSK} by approximating the underlying process for $S_t$ with a 
finite state Markov chain. The double Laplace transform of the Asian price 
with respect to strike and maturity can be expressed as a matrix multiplication. 
This method has been simplified further in \cite{Cui2018} by performing one Laplace
inversion analytically, which reduces the problem to the inversion of a single 
Laplace transform.

Another approach proposed is the PDE method \cite{RogersShi,Vecer,VecerXu}.
The Asian option pricing PDE can be solved either numerically, or using asymptotic 
expansion methods. We note the paper of \cite{FPP2013} which applied this 
method to derive precise expansions for Asian option prices in the local 
volatility model.  
Another very flexible method which can be used for models where 
$\log S_t$ has independent increments is a backward recursion combined with 
Fourier inversion methods. This was first proposed for the Black-Scholes (BS) model by 
Carverhill and Clewlow \cite{CarvClewlow}, and improved by Benhamou \cite{Benhamou}. 
The method was later extended to 
exponential L\'{e}vy models by Fusai \cite{Fusai2004} and Fusai and Meucci \cite{Fusai2008}.
An alternative
recursion method has been presented in \cite{PZIME} for pricing discrete sampled
Asian options in the BS model.

In \cite{PZLV} the 
authors obtained an asymptotic 
expansion for Asian options prices in the short-maturity limit $T\to 0$
in the local volatility model. Using large deviations 
theory methods \cite{Dembo,VaradhanII}, asymptotics for out-of-the money Asian options have been derived. 
For a general local volatility model the short maturity asymptotics 
is expressed in terms of a rate function which is given in terms of quadratures.
The result can be put in explicit form for the CEV model \cite{PZCEV}. 
The at-the-money (ATM) asymptotics has been obtained as well,
which is dominated by the fluctuations of the asset price around the spot price. 
We mention also the results of 
Gobet and Miri \cite{GobetMiri} who studied the expansion of the time-average of a diffusion in a small parameter such as time or volatility using Malliavin calculus, with applications to pricing Asian options. The implied volatility of Asian options 
in stochastic volatility models was studied recently in Al\`os et al. \cite{Alos2022} 
in a region of strikes close to the ATM point.
Also, Cai et al. \cite{CaiLiShi} presented closed-form expansions for Asian options with discrete time averaging in diffusion models.

One important question is the impact of jumps on the pricing of Asian options.
It is intuitively clear that jumps can have an important effect on the price of
an out-of-the-money (OTM) Asian options with small maturity, as the jump could possibly
take the option into in-the-money. Evidence for discontinuous behavior in the
dynamics of many financial assets has been presented in \cite{CarrWu}.
Since the economic motivation for trading
Asian options is to smooth out the impact of short term price fluctuations,
it is important to have a quantitative understanding of the jumps impact. 

The simplest jump-diffusion model is the Merton model \cite{Merton}, where the 
log-asset price $\log S_t$ is the sum of a standard Brownian motion plus
a compound Poisson process with constant intensity $\lambda$ and normally
distributed jump size. The pricing of Asian options on commodity futures
including stochastic volatility and jumps distributed according to the Merton
model has been presented in \cite{KyrQF}.

The impact of adding exactly one jump to the Black-Scholes
model on the pricing of Asian options has been studied in \cite{ChowLin}.
Another popular model is the double-exponential jump model of Kou \cite{Kou},
where the jump sizes can be both positive and negative, and are exponentially
distributed. 
American option pricing under this model has been studied in \cite{KouWang}, 
and pricing of Asian options has been presented in \cite{CaiKou}.

The model of \cite{Kou} has been generalized to the hyper-exponential model (HEM)
where the jump size distribution is a linear superposition of exponential 
distributions with the density:
\begin{equation}\label{HEMdef}
f_Y(x) = \sum_{i=1}^n p_i \eta_i e^{-\eta_i x} 1_{x\geq 0} +
\sum_{j=1}^n q_j \theta_j e^{\theta_j x} 1_{x < 0} \,.
\end{equation}
This model is widely used in the mathematical finance literature, and has been 
used for pricing options and exotic derivatives, such as barrier and American-type options 
by Boyarchenko and Levendorskii \cite{BL02} and Cai and Kou \cite{CaiKou} for Asian options.
By Bernstein's theorem the hyper-exponential 
distribution can approximate any L\'evy process with completely monotone 
density $\nu(dx)$.\footnote{The L\'evy density $\nu(dx) = \nu(x)dx$ is said to be 
completely monotone
if and only if, for all $k\geq 0$, one has $(-1)^k d\nu(x)/dx^k > 0$ for $x>0$, and
the same condition holds for $x<0$ with $\nu(-x)$. }
L\'{e}vy processes commonly used in finance for the CGMY, NIG and VG 
models are completely monotone. 

We comment also on the PDE method, which can be applied also in the presence 
of jumps \cite{Vecer,VecerXu}, when the PDE becomes an integro-differential 
equation. The smoothness properties
of the solution of this equation are not trivial, and they have been studied in
\cite{BayXing}.

We will study in this paper the impact of jumps on the short-maturity
asymptotics of Asian options. In contrast to the European options which have
been considered in the literature in the regime of small maturity, the
case of the Asian options has been less well studied.
The short-maturity asymptotics for European call and put options has been 
studied in the literature, under exponential L\'evy models 
and jump-diffusion models.

We summarize here a few results for the exponential L\'{e}vy models, and we refer to \cite{FLF} for
a more detailed list of references. 
See also \cite{Alos2007} for work on short-maturity expansions in jump-diffusion models using Malliavin calculus methods. 
Consider an exponential L\'evy model $S_t = S_{0}e^{X_t}$ where $X_t$
is a L\'evy process with density $\nu(dx)$ that starts at $0$ at time zero. In a wide class of
such models, the leading short-maturity asymptotics
of the European call options is given by \cite{BL02}:
\begin{equation}\label{EurSmallMat}
\lim_{T\to 0} \frac{1}{S_0 T} \mathbb{E}[(S_T - K)^+] = 
\int (e^x - (K/S_{0}))^+ \nu(dx)\,,\qquad K>S_{0}\,.
\end{equation}
A sufficient condition for this result to hold is $\int \min(|x|^2 ,1) e^x 
\nu(dx)<\infty$. See \cite{FL08} for weaker conditions. 
The next-to-leading order correction of $O(T^2)$ to the short-maturity asymptotics 
has been obtained in \cite{FLF}, and the leading asymptotics of the ATM skew was studied in \cite{FLO}.

In this paper we consider the pricing of Asian options in a wide class of
jump-diffusion models where the diffusion is given by a local volatility
model, and the jumps are modeled as compound Poisson processes (Section~\ref{sec:main}). 
This generalizes the Merton and doubly exponential models by allowing for a more general diffusion process. 
We consider the generalization to a general exponential L\'{e}vy process  in Section~\ref{sec:Levy}.
Asymptotics for floating strike Asian options are discussed in Section~\ref{sec:floating}.

A summary of the main results of the paper is given in Table~\ref{tab:payoffs}.
The table shows the instruments considered, their payoffs, and the leading short-maturity asymptotics for their prices in different regimes (OTM/ATM), under the models considered. 

Explicit asymptotic results are given for the Asian options prices in
the short maturity limit for some of the most popular models with jumps in the literature:
the Merton jump-diffusion model, the double-exponential jump diffusion model, and the Variance Gamma model (Section~\ref{sec:examples}). 

\begin{table}[htbp]
  \centering
  \caption{Summary of main results of the paper. The table shows the instruments considered, their payoff and the leading $T\to 0$ asymptotics of their price $V(T)$ with $V(T)=C(T)$ for call and $V(T)=P(T)$ for put options.
  The short-maturity asymptotics is different for OTM/ATM regimes.}
      \begin{tabular}{|l|cccc|}
    \hline
Instrument/Payoff & Jump model & $\lim_{T\to 0} T^{-p} V(T)$ & $p$ & Condition \\
    \hline
    \hline
\multicolumn{5}{|c|}{Fixed strike Asian options}  \\
\hline
Call $(A_T - K)^+$ & Compound Poisson & Theorem 2 (OTM) & 1 & Assump.~\ref{Assump1} \\
                                     &  & Theorem 4 (ATM) & 1/2 & Assump.~\ref{Assump1} \\
 & L\'evy jumps             & Theorem 8 (OTM)    & 1     & Assump.~\ref{assump2} \\
 \hline
Put $(K - A_T)^+$ & Compound Poisson & Theorem 3 (OTM) & 1 & Assump.~\ref{Assump1} \\
                                    & & Theorem 4 (ATM) & 1/2 & Assump.~\ref{Assump1} \\
 & L\'evy jumps & Theorem 9 (OTM) & 1 & Assump.~\ref{assump2} \\
\hline
\hline
\multicolumn{5}{|c|}{Floating strike Asian options}  \\
\hline
Call $(\kappa S_T - A_T)^+$ & L\'evy jumps & Theorem 12 (OTM) & 1 & Assump.~\ref{assump2} \\
                                & Compound Poisson & Theorem 14 (ATM) & 1/2 & Assump.~\ref{Assump1} \\
\hline
Put $(A_T - \kappa S_T)^+$ & L\'evy jumps & Theorem 13 (OTM) & 1 & Assump.~\ref{assump2} \\
                                & Compound Poisson & Theorem 14 (ATM) & 1/2 & Assump.~\ref{Assump1} \\
\hline
\end{tabular}%
  \label{tab:payoffs}%
\end{table}%

The asymptotic results are compared 
with independent MC simulations of Asian options in these models in Section~\ref{sec:numerical}. 
The technical proofs of all the theoretical results are provided in Appendix~\ref{sec:proofs}.

\section{Main Results}\label{sec:main}

Let $S_{t}$ be the asset price process. 
Assume that $r$ is the risk-free rate and $q$ is the dividend yield.
We are interested in the short-maturity limit as the maturity $T\rightarrow 0$.
In Pirjol and Zhu \cite{PZLV}, the asset price is assumed to follow a local volatility model:
\begin{equation}
\frac{dS_{t}}{S_{t}}=(r-q)dt
+\sigma(S_{t})dW_{t},
\qquad S_{0}>0.
\end{equation}
We assume that the local volatility function $\sigma(\cdot)$ satisfies
\begin{align}
&0<\underline{\sigma}\leq\sigma(\cdot)\leq\overline{\sigma}<\infty,\label{assumpI}
\\
&|\sigma(e^x)-\sigma(e^y)|\leq M|x-y|^{\alpha},\label{assumpII}
\end{align}
for some fixed $M,\alpha>0$ for any $x,y$ and $0<\underline{\sigma}<\overline{\sigma}<\infty$ are some fixed constants.
This is the same assumption used in Pirjol and Zhu \cite{PZLV}
to obtain the short maturity asymptotics for Asian options in local volatility models.
This assumption originally came from the paper of Varadhan \cite{Varadhan67} for short time interval diffusion processes. 

In this paper, we add the compound Poisson jumps independent of the diffusion part (and later we will extend our results
to allow L\'{e}vy jumps in Section~\ref{sec:Levy}), that is,
\begin{equation}\label{SJD}
S_{t}=\hat{S}_{t}e^{\sum_{i=1}^{N_{t}}Y_{i}},
\end{equation}
where $\hat{S}_{t}$ follows a local volatility model:
\begin{equation}\label{LV:SDE}
\frac{d\hat{S}_{t}}{\hat{S}_{t-}}=(r-q)dt
-\lambda\mu dt+\sigma(\hat{S}_{t})dW_{t},\qquad\hat{S}_{0}=S_{0},
\end{equation}
where $\sum_{i=1}^{N_{t}}Y_{i}$ is a compound Poisson process independent of $\hat{S}$ process, and
$\mu=\mathbb{E}[e^{Y_{1}}-1]$, and 
the jumps $Y_{i}$ are i.i.d. with the
probability distribution function $P(y)$, $-\infty<y<\infty$, that is, we allow
both positive and negative jumps and $Y_{i}$'s are independent of the Poisson process $N_{t}$
with an intensity $\lambda>0$.

Extensions of the local volatility model by adding jumps have been considered in the literature by
Andersen and Andreasen \cite{Andersen2000},
Benhamou et al. \cite{BenGobMiri} and
Pagliarani and Pascucci \cite{PagPasc}.

We will make the following assumption about the jump size distribution in the
compound Poisson process:
\begin{assumption}\label{Assump1}
The jumps $Y_1$ satisfy the condition
\begin{equation}\label{assumpIII}
\mathbb{E}[e^{\theta Y_{1}}]<\infty,
\qquad\text{for some $\theta>2$}.
\end{equation}
\end{assumption}

The price process $S_t$ can be written alternatively as
\begin{equation}
S_{t}=S_{0}e^{\int_{0}^{t}\sigma(\hat{S}_{s})dW_{s}
+(r-q-\lambda\mu)t
-\int_{0}^{t}\frac{1}{2}\sigma^{2}(\hat{S}_{s})ds
+\sum_{i=1}^{N_{t}}Y_{i}} .
\end{equation}

As the maturity $T\rightarrow 0$, we have $S_{T}\rightarrow S_{0}$ a.s.
We are interested in the asymptotics of
the out-of-the money Asian call option price:
\begin{equation}
C(T)=e^{-rT}\mathbb{E}\left[\left(A_{T}-K\right)^{+}\right],
\end{equation}
where $A(T)$ is defined in \eqref{A:T:defn}.

The $T\to 0$ asymptotics
of out-of-the-money Asian call options in the model where $S_t$ follows a local volatility model without jumps was studied
in Pirjol and Zhu \cite{PZLV}
and it is of the order $e^{-O(\frac{1}{T})}$. 
Similar results were obtained in \cite{Arguin2018}.
In the presence of jumps, the leading order will change.
The intuition behind this is that
the probability of one jump is of the order $O(T)$,
and one single jump can make the asset price
process back to the in-the-money region.
Since the probability of getting into the in-the-money region
from the contributions of the diffusions is 
of the order $e^{-O(\frac{1}{T})}$, it is intuitively
clear that in the presence of jumps, $C(T)=O(T)$. 
So the key question here is to determine the precise
leading order asymptotics for $C(T)$ 
as the maturity $T\rightarrow 0$.

We have the following main result:

\begin{theorem}\label{ThmI}
Assume that the asset price $S_t$ follows the process (\ref{SJD}) and the jump distribution satisfies Assumption~\ref{Assump1}.
Then the leading order asymptotics for out-of-the-money, i.e. $K>S_{0}$, short maturity Asian call options
is given by
\begin{equation}\label{aC}
\lim_{T\rightarrow 0}\frac{C(T)}{T}
=\lambda\int_{0}^{1}\int_{\log(\frac{K-S_{0}t}{S_{0}(1-t)})}^{\infty}\left(S_{0}t+S_{0}e^{y}(1-t)-K\right)dP(y)dt,
\end{equation}
provided that the limit is positive.
\end{theorem}

Similarly, we have the following result for OTM Asian put options:

\begin{theorem}\label{ThmPut}
Assume that the asset price $S_t$ follows the process (\ref{SJD}) and the jump distribution satisfies Assumption~\ref{Assump1}.
Then the leading order asymptotics for out-of-the-money, i.e. $K<S_{0}$, short maturity Asian put options is given by
\begin{equation}\label{aP}
\lim_{T\rightarrow 0}\frac{P(T)}{T}
=\lambda\int_{0}^{\frac{K}{S_{0}}}\int_{-\infty}^{\log(\frac{K-S_{0}t}{S_{0}(1-t)})}\left(K-S_{0}t-S_{0}e^{y}(1-t)\right)dP(y)dt,
\end{equation}
provided that the limit is positive.
\end{theorem}

Next, we study the short-maturity asymptotics for ATM Asian options. 
The key observation is that for ATM Asian options,
a jump occurs with probability $O(T)$, whereas
the contribution from the local volatility (diffusion) part of the asset price \eqref{SJD}
has order $O(\sqrt{T})$ from the result
in Pirjol and Zhu \cite{PZLV} for local volatility models (without jumps),
and hence the leading-order short-maturity asymptotics for ATM Asian options 
is provided by the local volatility (diffusion) part, instead of the jumps
as in the OTM case (Theorem~\ref{ThmI} and Theorem~\ref{ThmPut}). 
In particular, we have the following result.

\begin{theorem}\label{ThmIII}
Assume that the asset price $S_t$ follows the process (\ref{SJD}) and the jump distribution satisfies Assumption~\ref{Assump1}.
The leading order asymptotics for at-the-money, i.e. $K=S_{0}$, short maturity Asian call and put options
is given by
\begin{equation}
\lim_{T\rightarrow 0}\frac{C(T)}{\sqrt{T}}
=\lim_{T\rightarrow 0}\frac{P(T)}{\sqrt{T}}
=\frac{1}{\sqrt{6\pi}}\sigma(S_{0})S_{0}.
\end{equation}
\end{theorem}

\begin{remark}
Denoting $C_{E}$ and $P_{E}$ the prices of European call
and put options respectively, the leading order asymptotic result with
the underlying asset price process \eqref{SJD} is given by \cite{BL02}:
\begin{align}
&\lim_{T\rightarrow 0}\frac{C_{E}(T)}{T}
=\lambda\int_{-\infty}^{\infty}(S_{0}e^{y}-K)^{+}dP(y),
\qquad\text{for $S_{0}<K$},
\\
&\lim_{T\rightarrow 0}\frac{P_{E}(T)}{T}
=\lambda\int_{-\infty}^{\infty}(K-S_{0}e^{y})^{+}dP(y),
\qquad\text{for $S_{0}>K$},
\end{align}
and for $S_{0}=K$ we have \cite{MN2011}:
\begin{equation}
\lim_{T\rightarrow 0}\frac{C_{E}(T)}{\sqrt{T}}
=\lim_{T\rightarrow 0}\frac{P_{E}(T)}{\sqrt{T}}
=\frac{1}{\sqrt{2\pi}}\sigma(S_{0})S_{0} \,.
\end{equation}
\end{remark}

\begin{remark}
We can link the asymptotics for the Asian options
and that of the European options with the underlying asset price process \eqref{SJD}
as follows:
\begin{align}
&\lim_{T\rightarrow 0}\frac{C(T)}{C_{E}(T)}
=\frac{\int_{0}^{1}\int_{\log(\frac{K-S_{0}t}{S_{0}(1-t)})}^{\infty}\left(S_{0}t+S_{0}e^{y}(1-t)-K\right)dP(y)dt}
{\int_{\log(K/S_{0})}^{\infty}(S_{0}e^{y}-K)dP(y)},
\qquad\text{for $S_{0}<K$},
\\
&\lim_{T\rightarrow 0}\frac{P(T)}{P_{E}(T)}
=\frac{\int_{0}^{\frac{K}{S_{0}}}\int_{-\infty}^{\log(\frac{K-S_{0}t}{S_{0}(1-t)})}\left(K-S_{0}t-S_{0}e^{y}(1-t)\right)dP(y)dt}{\int_{-\infty}^{\log(K/S_{0})}(K-S_{0}e^{y})dP(y)},
\qquad\text{for $S_{0}>K$},
\\
&\lim_{T\rightarrow 0}\frac{C(T)}{C_{E}(T)}
=\lim_{T\rightarrow 0}\frac{P(T)}{P_{E}(T)}
=\frac{1}{\sqrt{3}},
\qquad\text{for $S_{0}=K$}.
\end{align}
\end{remark}


\section{Local Volatility Models with L\'{e}vy Jumps}\label{sec:Levy}

The compound Poisson jumps in our model \eqref{SJD}-\eqref{LV:SDE} can be generalized
to L\'{e}vy jumps, so that our extended model can include the exponential L\'{e}vy model as a special case.
Let $X_{t}$ be a L\'{e}vy process with the triplet $(0,b,\nu)$, that is, 
$X_{t}$ admits the decomposition:
\begin{equation}
X_{t}=bt+\int_{0}^{t}\int_{|x|\leq 1}x(\mu-\bar{\mu})(dx,ds)
+\int_{0}^{t}\int_{|x|>1}x\mu(dx,ds),
\end{equation}
where $\mu$ is a Poisson measure on $\mathbb{R}_{+}\times\mathbb{R}\backslash\{0\}$
with mean measure $\bar{\mu}(dx,dt)=\nu(dx)dt$. 

We will require the following assumption about the L\'evy measure.
\begin{assumption}\label{assump2}
The L\'evy measure of the process $X_t$ satisfies the condition
\begin{equation}
\int(x^{2}\wedge 1)e^{\theta x}\nu(dx)<\infty,
\end{equation}
for some $\theta>2$.
\end{assumption}

We define the asset price process as
\begin{equation}\label{SLevy}
S_{t}=\hat{S}_{t}e^{X_{t} },
\end{equation}
where $\hat S_t$ satisfies a local volatility model without jumps: 
\begin{equation}
\frac{d\hat{S}_{t}}{\hat{S}_{t-}}=(r-q)dt
-\mu dt+\sigma(\hat{S}_{t})dW_{t},\qquad\hat{S}_{0}=S_{0},
\end{equation}
and the $X_{t}$ process is independent of the $\hat{S}_{t}$ process, 
and $\mu= \psi(-i)$ is the compensator, where $\psi(u)$ is defined as 
$\mathbb{E}[e^{i u X_1}] = e^{\psi(u)}$.
The asset price model \eqref{SLevy} covers a wide range of models used in finance
such as Merton jump diffusion model \cite{Merton}, double exponential jump model \cite{Kou}
and Variance Gamma model \cite{Madan1990}.

We have the following short-maturity asymptotics for Asian options under the model (\ref{SLevy}).

\begin{theorem}\label{Thm:LevyC}
The leading order asymptotics for out-of-the-money, i.e. $K>S_{0}$, 
short maturity Asian call options under the model (\ref{SLevy})
with the Assumption~\ref{assump2}, is given by
\begin{equation}
\lim_{T\rightarrow 0}\frac{C(T)}{T}
=\int_{0}^{1}\int_{\log(\frac{K-S_{0}t}{S_{0}(1-t)})}^{\infty}\left(S_{0}t+S_{0}e^{y}(1-t)-K\right)\nu(dy)dt.
\end{equation}
\end{theorem}

Similarly, we have the following result for OTM Asian put options:

\begin{theorem}\label{Thm:LevyP}
The leading order asymptotics for out-of-the-money, i.e. $K<S_{0}$, 
short maturity Asian put options under the model (\ref{SLevy})
with the Assumption~\ref{assump2}, is given by
\begin{equation}
\lim_{T\rightarrow 0}\frac{P(T)}{T}
=\int_{0}^{\frac{K}{S_{0}}}\int_{-\infty}^{\log(\frac{K-S_{0}t}{S_{0}(1-t)})}\left(K-S_{0}t-S_{0}e^{y}(1-t)\right)\nu(dy)dt.
\end{equation}
\end{theorem}


Unlike the underlying model with compound Poisson jumps, 
we will not pursue the ATM asymptotics here. 
In the presence of L\'{e}vy jumps, it is known 
that the short maturity asymptotics
of ATM European options
can exhibit various
order dependence on $T$ depending on the L\'{e}vy measure \cite{MN2011}.
As a result, we do not anticipate that an unified result
similar to Theorem~\ref{ThmIII} holds for this case, and the ATM asymptotics
for Asian options in the presence of L\'{e}vy jumps
will be left as a future research direction.

\begin{remark}
Denoting $C_{E}$ and $P_{E}$ the prices of European call
and put options respectively, the leading order asymptotics with
the underlying asset price process \eqref{SLevy} are given by \cite{BL02}:
\begin{align}
&\lim_{T\rightarrow 0}\frac{C_{E}(T)}{T}
=\int_{-\infty}^{\infty}(S_{0}e^{y}-K)^{+}\nu(dy),
\qquad\text{for $S_{0}<K$},
\\
&\lim_{T\rightarrow 0}\frac{P_{E}(T)}{T}
=\int_{-\infty}^{\infty}(K-S_{0}e^{y})^{+}\nu(dy),
\qquad\text{for $S_{0}>K$}.
\end{align}
\end{remark}

\begin{remark}
We can link the asymptotics for the Asian options
and that of the European options with the underlying asset price process \eqref{SLevy}
as follows:
\begin{align}
&\lim_{T\rightarrow 0}\frac{C(T)}{C_{E}(T)}
=\frac{\int_{0}^{1}\int_{\log(\frac{K-S_{0}t}{S_{0}(1-t)})}^{\infty}\left(S_{0}t+S_{0}e^{y}(1-t)-K\right)\nu(dy)dt}
{\int_{\log(K/S_{0})}^{\infty}(S_{0}e^{y}-K)\nu(dy)},
\qquad\text{for $S_{0}<K$},
\\
&\lim_{T\rightarrow 0}\frac{P(T)}{P_{E}(T)}
=\frac{\int_{0}^{\frac{K}{S_{0}}}\int_{-\infty}^{\log(\frac{K-S_{0}t}{S_{0}(1-t)})}\left(K-S_{0}t-S_{0}e^{y}(1-t)\right)\nu(dy)dt}{\int_{-\infty}^{\log(K/S_{0})}(K-S_{0}e^{y})\nu(dy)},
\qquad\text{for $S_{0}>K$}.
\end{align}
\end{remark}


\section{Floating Strike Asian Options}\label{sec:floating}

There are many variations of the standard Asian options in the finance literature
and one of the most used is the so-called floating strike Asian options.
The price of the floating strike Asian call/put options are given by
\begin{align}
&C_f(T):=e^{-rT}\mathbb{E}\left[\left(\kappa S_{T}-A_{T}\right)^{+}\right],
\\
&P_f(T):=e^{-rT}\mathbb{E}\left[\left(A_{T}-\kappa S_{T}\right)^{+}\right],
\end{align}
where $A_{T}$ is defined in \eqref{A:T:defn} and $\kappa>0$ is the strike, see e.g.
\cite{Alziary,HW,RogersShi}.
The floating-strike Asian option is more difficult to price than the fixed-strike case
because the joint law of $S_{T}$ and $A_{T}$ is needed.
When $\kappa<1$, the call option is OTM, the put option is ITM;
when $\kappa>1$, the call option is ITM, the put option is OTM;
when $\kappa=1$, the call/put options are ATM. We are interested
in the short maturity, i.e. $T\rightarrow 0$ asymptotics of these options.

We assume that asset price follows \eqref{SLevy}, which is a local volatility
model with L\'{e}vy jumps. Hence, we extend the short-maturity asymptotics
for Asian options in \cite{PZLV} that considers the local volatility model (without jumps).
For the Black-Scholes model, it was shown by Henderson and 
Wojakowski \cite{HW} that the floating-strike Asian options with continuous 
time averaging can be related to fixed strike ones. 
However, in our general setting of local volatility models with L\'{e}vy jumps, the equivalence
relations obtained in \cite{HW} do not hold, and hence the asymptotics for floating 
strike Asian options must be obtained independently from that of the fixed 
strike Asian options.

We have the following short-maturity asymptotics for the floating strike Asian call options.

\begin{theorem}\label{Thm:LevyC:floating}
The leading order asymptotics for out-of-the-money, i.e. $\kappa<1$, 
short maturity Asian floating strike call options under the model (\ref{SLevy})
with Assumption~\ref{assump2}, is given by
\begin{equation}
\lim_{T\rightarrow 0}\frac{C_{f}(T)}{T}
=\int_{1-\kappa}^{1}\int_{\log(\frac{t}{\kappa-(1-t)})}^{\infty}S_{0}\left(\kappa e^{y}-t-e^{y}(1-t)\right)\nu(dy)dt\,.
\end{equation}
\end{theorem}

Similarly, we have the following result for the floating strike Asian put options.

\begin{theorem}\label{Thm:LevyP:floating}
The leading order asymptotics for out-of-the-money, i.e. $\kappa>1$, 
short maturity Asian floating strike put options under the model (\ref{SLevy})
with Assumption~\ref{assump2}, is given by
\begin{equation}
\lim_{T\rightarrow 0}\frac{P_{f}(T)}{T}
=\int_{0}^{1}\int_{-\infty}^{\log(\frac{t}{\kappa-(1-t)})}S_{0}\left(t+e^{y}(1-t)-\kappa e^{y}\right)\nu(dy)dt\,.
\end{equation}
\end{theorem}

In the special case of the local volatility model with compound Poisson jumps \eqref{SJD},
we have the following ATM asymptotics result which is an analogue of Theorem~\ref{ThmIII}.

\begin{theorem}\label{Thm:ATM:floating}
The leading order asymptotics for at-the-money, i.e. $\kappa=1$, short maturity Asian floating strike call and put options
under the model \eqref{SJD} with Assumption~\ref{Assump1} is given by
\begin{equation}
\lim_{T\rightarrow 0}\frac{C_{f}(T)}{\sqrt{T}}
=\lim_{T\rightarrow 0}\frac{P_{f}(T)}{\sqrt{T}}
=\frac{1}{\sqrt{6\pi}}\sigma(S_{0})S_{0}.
\end{equation}
\end{theorem}

For the convenience of the reader, the results of this section and the previous section 
are summarized in Table~\ref{tab:payoffs}. 
The table shows the instruments considered, their payoff (i.e. call or put options, fixed or floating strikes)
and the leading $T\to 0$ asymptotics of their prices $V(T)$ (i.e. $V(T)=C(T)$ for fixed-strike call options,
$V(T)=P(T)$ for fixed-strike put options, $V(T)=C_{f}(T)$ for floating-strike call options, 
$V(T)=P_{f}(T)$ for floating-strike put options), the underlying asset price dynamics (i.e. with compound Poisson
or L\'{e}vy jumps), the regime (i.e. OTM or ATM), and the conditions that are needed (i.e. Assumption~\ref{Assump1} or Assumption~\ref{assump2}).

\section{L\'{e}vy Jumps Models in Finance}\label{sec:examples}

In this section, let us consider several popular models in option pricing in finance
that involve L\'{e}vy jumps. In each example, we compute explicitly
the limit presented in Theorem~\ref{Thm:LevyC} and Theorem~\ref{Thm:LevyP}.

\subsection{Merton Jump Diffusion Model}

The Merton jump diffusion model \cite{Merton} is defined by 
\begin{equation}\label{merton:model}
S_t = S_0 e^{\sigma W_t + X_t + (r - q - \frac12 \sigma^2) t - \lambda \mu t}\,,
\end{equation}
where $X_t = \sum_{i=1}^{N_t} Y_i$ is a compound Poisson process, where $N_t$ is a Poisson process is intensity $\lambda>0$, and $Y_{i}$'s are i.i.d. normally distributed
jumps $Y_1 =\delta Z + \alpha$ with $Z=N(0,1)$ 
that corresponds to a jump density 
\begin{equation}\label{dPMerton}
dP(y) = \frac{1}{\sqrt{2\pi}\delta} e^{-\frac{1}{2\delta^2} (y-\alpha)^2}dy\,,
\end{equation}
independent of the Poisson process $N_{t}$ and the compensator $\mu$ is
\begin{equation}
\mu =\mathbb{E}[e^{Y_1}  - 1] = e^{\alpha + \frac12 \delta^2} - 1.
\end{equation}
We can rewrite \eqref{merton:model} as
$S_{t}=\hat{S}_{t}e^{X_t}$, 
with $\hat{S}_{t}=S_{0}e^{\sigma W_t + (r - q - \frac12 \sigma^2-\lambda\mu) t }$
so that the model \eqref{merton:model} is covered by \eqref{SLevy}.
We note that the Assumption~\ref{Assump1} is satisfied for any values of
the $\alpha,\delta$ parameters. 

\begin{proposition}\label{prop:Merton}
Under the Merton jump-diffusion model, we have the
leading order short maturity asymptotics for OTM Asian options.

i) For an OTM Asian call option $(K>S_0)$ we have
\begin{equation}
\lim_{T\to 0} \frac{C(T)}{T} = 
\lambda \int_0^1 \left[ (S_0 t-K) I_1(t) + S_0(1-t) I_2(t) \right] dt\,.
\end{equation}

ii) For an OTM Asian put option $(K<S_0)$ we have
\begin{equation}
\lim_{T\to 0} \frac{P(T)}{T} = 
\lambda \int_0^{K/S_0} \left[ (K - S_0 t) I_3(t) - S_0(1-t) I_4(t) \right] dt\,,
\end{equation}
where we denoted:
\begin{align}
&I_1(t) := \Phi\left( -\frac{1}{\delta} \log \left[ \frac{K-S_0 t}{S_0(1-t)} e^{-\alpha} \right] \right), \\
&I_2(t) := e^{\alpha+\frac12\delta^2} \Phi\left( -\frac{1}{\delta} \log \left[ \frac{K-S_0 t}{S_0(1-t)} e^{-\alpha-\delta^2} \right] \right),
\end{align}
and
\begin{align}
&I_3(t) := \Phi\left( \frac{1}{\delta} \log \left[ \frac{K-S_0 t}{S_0(1-t)} e^{-\alpha} \right] \right) 
 = 1 - I_1(t),\\
&I_4(t) := 
e^{\alpha+\frac12\delta^2} \Phi\left( \frac{1}{\delta} 
\log \left[ \frac{K-S_0 t}{S_0(1-t)} e^{-\alpha-\delta^2} \right] \right) =
e^{\alpha+\frac12\delta^2} - I_2(t)\,,
\end{align}
where $\Phi(\cdot )$ is the cumulative distribution function of $N(0,1)$.
\end{proposition}

We give next also the explicit results for the floating strike Asian options in the Merton model following from Theorems~\ref{Thm:LevyC:floating}
and \ref{Thm:LevyP:floating}.

\begin{proposition}\label{prop:floatMerton}
Under the Merton jump-diffusion model, we have the
leading order short maturity asymptotics for OTM floating strike Asian options.

i) For an OTM floating strike Asian call option $(\kappa < 1)$ we have
\begin{equation}
\lim_{T\to 0} \frac{C_f(T)}{T} = 
\lambda S_0 \int_{1-\kappa}^1 \left[ -t  I_{1,f}(t) + (\kappa - (1-t)) I_{2,f}(t) \right] dt\,.
\end{equation}

ii) For an OTM floating strike Asian put option $(\kappa>1)$ we have
\begin{equation}
\lim_{T\to 0} \frac{P_f(T)}{T} = 
\lambda S_0 \int_0^{1} \left[ t I_{3,f}(t) + (1-t-\kappa) I_{4,f}(t) \right] dt\,,
\end{equation}
where
\begin{align}
&I_{1,f}(t) := \Phi\left( \frac{- \log y_0(t) + \alpha}{\delta}  \right), \\
&I_{2,f}(t) := e^{\alpha+\frac12\delta^2} \Phi\left( \frac{-\log y_0(t) + \alpha + \delta}{\delta} \right),
\end{align}
and
\begin{align}
&I_{3,f}(t) := \Phi\left( \frac{\log y_0(t) - \alpha}{\delta} \right) 
 = 1 - I_{1,f}(t),\\
&I_{4,f}(t) := 
e^{\alpha+\frac12\delta^2} \Phi\left( \frac{\log y_0(t) - \alpha - \delta^2}{\delta} \right) =
e^{\alpha+\frac12\delta^2} - I_{2,f}(t)\,,
\end{align}
with $y_0(t) := \frac{t}{\kappa-(1-t)}$ and $\Phi(\cdot )$ is the cumulative distribution function of $N(0,1)$.
\end{proposition}

\subsection{Compound Poisson with Double Exponential Jumps}

We consider next the jump diffusion model (see e.g. Kou \cite{Kou}
and Kou and Wang \cite{KouWang}):
\begin{equation}\label{double:model}
S_{t}= \hat S_{t}e^{X_{t}},
\end{equation}
where $\hat S_t = S_0 e^{\sigma W_t + (r-q -\frac12\sigma^2-\lambda\mu)t}$ is a 
geometric Brownian motion and $X_{t}=\sum_{i=1}^{N_{t}}Y_{i}$ is a compound Poisson process
where the compound Poisson jumps 
arrive at rate $\lambda$, and the jump sizes follow the double exponential 
distribution with the probability density function:
\begin{equation}\label{doubleExp}
dP(y)=\alpha\eta_{1}e^{-\eta_{1}y}1_{y\geq 0}dy+(1-\alpha)\eta_{2}e^{\eta_{2}y}1_{y<0}dy,
\end{equation}
where $0\leq\alpha\leq 1$ and $\eta_{1}>1$, $\eta_2>0$.
The compensator is given by:
\begin{equation}
\mu = \mathbb{E}[e^{Y_1}]-1 = \alpha \frac{\eta_1}{\eta_1-1} + (1-\alpha) \frac{\eta_2}{\eta_2+1} - 1\,.
\end{equation}
The model \eqref{double:model} is covered by \eqref{SLevy}.

\begin{proposition}\label{prop:Kou}
We have the following short-maturity asymptotics for the OTM Asian options in the
Double Exponential jump model (\ref{double:model}).
Assume $\eta_{1}, \eta_{2} > 2$ which ensures that the  
Assumption~\ref{Assump1} is satisfied.

i) For OTM Asian call options, i.e. $K>S_{0}$, we have
\begin{equation}
\lim_{T\rightarrow 0}\frac{C(T)}{T}
=\frac{\lambda\alpha S_{0}}{\eta_{1}^{2}-1}\left(\frac{K}{S_{0}}\right)^{1-\eta_{1}}
{}_{2}F_{1}\left(1,\eta_{1}-1;\eta_{1}+2;\frac{S_{0}}{K}\right),
\end{equation}
where $_{2}F_{1}(a,b;c;z)$ is a hypergeometric function.

ii) For OTM Asian put options, i.e. $K<S_{0}$, we have
\begin{align}
\lim_{T\rightarrow 0}\frac{P(T)}{T}
&=\frac{\lambda(1-\alpha)S_0}{(\eta_{2}+1)(\eta_{2}+1)}\left(\frac{K}{S_{0}}\right)^{\eta_{2}+2}
{}_{2}F_{1}\left(1,\eta_{2};\eta_{2}+3;\frac{K}{S_{0}}\right).
\end{align}
\end{proposition}

\begin{remark}
The distribution (\ref{doubleExp}) is a particular case of the hyperexponential 
distribution (\ref{HEMdef})
from which it follows by taking $n=1$ and $p_1 =\alpha, q_1 = 1-\alpha$. 
We present here only the results with $n=1$, the generalization to arbitrary 
$n\in\mathbb{N}$ is immediate. 
\end{remark}

\begin{remark}
When $p(y)$ is the probability density function of the double exponential 
distribution defined in \eqref{doubleExp}, we have the short maturity 
asymptotics for European options
\begin{equation}
\lim_{T\rightarrow 0}\frac{C_{E}(T)}{T}
=\frac{\lambda\alpha}{\eta_{1}-1}S_{0}\left(\frac{S_{0}}{K}\right)^{\eta_{1}-1},
\qquad\text{for $S_{0}<K$},
\end{equation}
and
\begin{equation}
\lim_{T\rightarrow 0}\frac{P_{E}(T)}{T}
=\frac{\lambda(1-\alpha)}{\eta_{2}+1}K\left(\frac{K}{S_{0}}\right)^{\eta_{2}},
\qquad\text{for $S_{0}>K$}.
\end{equation}
\end{remark}


\subsection{Variance Gamma Process}

The Variance Gamma (VG) process was introduced in 
financial modeling in \cite{Madan1990}, and its application to option pricing
was explored in \cite{Madan1998}. 
The VG process is defined as 
\begin{equation}\label{VGdef}
X^{\mathrm{VG}}_t = \sigma W_{\Gamma_t(\nu)} + \theta \Gamma_t(\nu)\,,
\end{equation}
where $W_t$ is a standard Brownian motion and the subordinator
$\Gamma_t(\nu)$ is a Gamma process with mean rate 1 and variance rate
$\nu$. Its increments $\Delta \Gamma_t(\nu) := \Gamma_{t+\Delta t}(\nu) - \Gamma_t(\nu)$ are Gamma distributed with
mean $\Delta t$ and variance $\nu \Delta t$.

The characteristic function of the VG process is 
\begin{equation}
\phi_{\mathrm{VG}}(u) = \mathbb{E}\left[e^{iu X_1^{\mathrm{VG}}}\right] = \left(
1 - i u \theta\nu + \frac12 \sigma^2 u^2\nu \right)^{-\frac{1}{\nu}} \,.
\end{equation}

The asset price in the VG model has the form
\begin{equation}\label{VGmodel}
S_t = S_0 e^{X_t^{\mathrm{VG}} + (r - q - \mu) t }\,,
\end{equation}
where the compensator $\mu$ is
\begin{equation}
\mu = \log \phi_{\mathrm{VG}}(-i) = - \frac{1}{\nu} \log \left(
1 - \left(\theta + \frac12 \sigma^2\right) \nu \right) \,.
\end{equation}
Since the VG process is a pure jump process, 
the asset price model \eqref{VGmodel} is covered by \eqref{SLevy}.

Following \cite{Carr2002} we introduce the notations
\begin{equation}
1 - i u \theta\nu + \frac12 \sigma^2 u^2\nu = (1 - i\eta_p u)(1 - i\eta_n u)\,,
\end{equation}
with
\begin{equation}
\eta_n :=  \sqrt{\frac14 \theta^2 \nu^2 + \frac12 \sigma^2 \nu } - \frac12\theta \nu\,,
\qquad
\eta_p :=  \sqrt{\frac14 \theta^2 \nu^2 + \frac12 \sigma^2 \nu } + \frac12\theta \nu\,.
\label{etap}
\end{equation}

The L\'evy density of the VG process is $\nu_{\mathrm{VG}}(dx) = \nu_{\mathrm{VG}}(x) dx$, with
\begin{equation}\label{VGmeasure}
\nu_{\mathrm{VG}}(x) = \frac{1}{\nu |x|} e^{-\frac{1}{\eta_n} |x| } 1_{x<0} + 
\frac{1}{\nu x} e^{- \frac{1}{\eta_p} x} 1_{x>0}\,.
\end{equation}

This is a special case of the L\'evy density for the CGMY model \cite{Carr2002} 
which has the form
\begin{equation}\label{CGMYmeasure}
\nu_{\mathrm{CGMY}}(x) = \frac{C}{| x|^{1+Y}} e^{-G |x|} 1_{x<0} + 
\frac{C}{x^{1+Y}} e^{-M x} 1_{x>0}\,,
\end{equation}
where $C,M,G,Y$ are real parameters. This reduces to (\ref{VGmeasure}) by taking
\begin{equation}\label{CGMY2VG}
C = \frac{1}{\nu} \,,\quad
G = \frac{1}{\eta_n}\,,\quad 
M = \frac{1}{\eta_p}\,,\quad 
Y = 0\,.
\end{equation}

\begin{lemma}\label{lem:VG}
Assumption~\ref{assump2} holds for the CGMY density (\ref{CGMYmeasure}) provided that $M>2$ and $Y<2$. For the VG model this reduces to the condition
\begin{equation}\label{VGcond}
2(\theta + \sigma^2) \nu < 1 \,.
\end{equation}
\end{lemma}

We are now in a position to state the short-maturity asymptotics for the Asian options in the VG model. 

\begin{proposition}\label{prop:VG}
Assume that the asset price $S_t$ follows the VG model (\ref{VGmodel})
with parameters 
$(\theta, \nu,\sigma)$ which satisfy the condition (\ref{VGcond}).
For simplicity we use the CGMY notations for the VG parameters.
We have the following short-maturity asymptotics for OTM Asian options.

i) The short-maturity asymptotics for OTM Asian call option ($K>S_0$) is
\begin{equation}
\lim_{T\to 0} \frac{C(T)}{T} = \int_0^1 ( (S_0 t - K) J(t,M) + S_0(1-t) J(t,M-1))dt \,,
\end{equation}
with
\begin{equation}\label{Jres}
J(t,\alpha) := C \int_{\log (\frac{K/S_0-t}{1-t})}^\infty e^{-\alpha y} \frac{dy}{y} = 
- C \cdot \mathrm{li}\left( \left( \frac{S_0(1-t)}{K-S_0 t} \right)^\alpha \right)\,,\quad \alpha>0\,.
\end{equation}
The argument of the $\mathrm{li}(z)$ function takes values in $[0,1]$ for any $t\in [0,1]$.

ii)
The short-maturity asymptotics for an OTM Asian put option $(K < S_0)$ is

\begin{equation}
\lim_{T\to 0} \frac{P(T)}{T} = \int_0^{K/S_0}  ( (K - S_0 t ) \tilde J(t,G) - 
S_0(1-t) \tilde J(t,G+1))dt\,.
\end{equation}
with
\begin{equation}\label{Jtres}
\tilde J(t,\alpha) := C \int_{-\log (\frac{K/S_0-t}{1-t})}^\infty e^{-\alpha y} \frac{dy}{y} = 
- C \cdot \mathrm{li}\left( \left( \frac{K-S_0 t}{S_0(1-t)} \right)^\alpha \right)\,,\quad \alpha>0\,.
\end{equation}
The argument of the $\mathrm{li}(z)$ function takes values in $[0,1]$ for any
 $t\in [0,K/S_0]$. 
\end{proposition}

The function $\texttt{li}(z)$ is the so-called logarithmic integral and is
defined as
\begin{equation}
\mathrm{li}(z) := \int_0^z \frac{dt}{\log t}\,.
\end{equation}
This is available for example in \textit{Mathematica} as \texttt{LogIntegral[z]}.

The function $\mathrm{li}(z)$ has the following properties: 
\begin{itemize}
\item[i)] $\lim_{z\to 0} \mathrm{li}(z) = 0$. 
\item[ii)] For $0 < z < 1$ it is a negative and decreasing function, approaching $-\infty$ as $z\to 1$.
\item[iii)] For $z > 1$ it is an increasing function; it vanishes at $z \simeq 1.4$, and then goes to infinity as $z\to \infty$.
\end{itemize}

The ATM limits of the short-maturity coefficients can be evaluated in closed form.
Denoting the two limits in Proposition~\ref{prop:VG} as $a_C^{\mathrm{VG}}(K/S_0), a_P^{\mathrm{VG}}(K/S_0)$, we have
\begin{equation}
\lim_{K/S_0\searrow 1} a_C^{\mathrm{VG}}(K/S_0) = C \cdot \mathrm{arctanh}\left(\frac{1}{2M-1}\right)\,,
\end{equation}
and 
\begin{equation}
\lim_{K/S_0\nearrow 1} a_P^{\mathrm{VG}}(K/S_0) = C \cdot \mathrm{arctanh}\left(\frac{1}{2G+1}\right)\,.
\end{equation}
However, we emphasize that these asymptotic results apply only for $K\neq S_0$ and do not extend to the ATM case $K=S_0$.


\section{Analytical Approximations and Numerical Tests}\label{sec:numerical}

We propose in this section an analytical approximation for Asian options in jump-diffusion models, which is consistent with both the OTM and ATM short maturity asymptotics derived in this paper.
This analytical approximation will be tested by comparing with numerical simulations for a few commonly used jump-diffusion models.

i) For $K\geq S_0$, the Asian call option is represented
as the sum of a diffusive component and the jump component
\begin{equation}\label{appC}
C(K,T) = C_{\rm diff}(K,T; \Sigma_{\mathrm{LN}}(K/S_0)) + a_C(K) T\,.
\end{equation}
The put option is obtained from put-call parity
\begin{equation}\label{PCAsian}
C(K,T) - P(K,T) = e^{-rT}(A(T) - K) \,,
\end{equation}
where the forward of the average asset price is 
\begin{equation}\label{Afwddef}
A(T) = S_0 \frac{1}{(r-q)T} (e^{(r-q)T} - 1 ) \,.
\end{equation}

ii) For $K\leq S_0$, the Asian put option is represented in a similar way as\begin{equation}\label{appP}
P(K,T) = P_{\rm diff}(K,T) + a_P(K) T\,,
\end{equation}
and the Asian call put option is obtained from put-call parity (\ref{PCAsian}).

The coefficients $a_C(K),a_P(K)$ are defined as
\begin{equation}\label{JumpLim}
a_C(K) := \lim_{T\to 0} \frac{C(K,T)}{T} \,,\quad
a_P(K) := \lim_{T\to 0} \frac{P(K,T)}{T}\,,
\end{equation}
and are given in Proposition~\ref{prop:Merton} for the Merton model,
Proposition~\ref{prop:Kou} for the double exponential jump-diffusion model,
and in Proposition~\ref{prop:VG} for the VG model.

The diffusive component has the Black-Scholes form:
\begin{align}
& C_{\rm diff}(K,T) = e^{-r T}[A(T) \Phi(d_1) - K \Phi(d_2) ]\,, \\
& P_{\rm diff}(K,T) = e^{-r T}[K \Phi( - d_2) - A(T) \Phi(-d_1) ]\,,
\end{align}
where $A(T)$ is given in (\ref{Afwddef}) and 
$d_{1,2}=\frac{1}{\Sigma_{\rm LN}(K)\sqrt{T}}
\left( \log\frac{A(T)}{K} \pm \frac12 \Sigma^2_{\rm LN}(K) T \right)
$. The function $\Sigma_{\mathrm{LN}}(K)$ is the equivalent log-normal volatility of an Asian option in the local volatility model with volatility $\sigma(\cdot )$. This function was introduced and studied in \cite{PZLV}. 

The precise form of $\Sigma_{\rm LN}(K)$ depends on the local volatility function $\sigma(\cdot)$.
For the Black-Scholes model the function $\Sigma_{\mathrm{LN}}(K)$ depends only on moneyness $k=K/S_0$ and is given in closed form in Proposition~18 of \cite{PZLV}. 
The first few terms in its expansion around the ATM point are
\begin{equation}\label{SigLN}
\Sigma_{\mathrm{LN}}(K) = \frac{1}{\sqrt3} \sigma \left( 1 + \frac{1}{10} \log k -
\frac{23}{2100} \log^2 k + \frac{1}{3500} \log^3 k + O\left(\log^4 k\right)\right) \,.
\end{equation}
This series expansion is convenient for practical evaluations for $|\log(K/S_0)|$ not too large. 

It is easy to check that the approximation (\ref{appC}), (\ref{appP}) reproduces the short-maturity limit in Theorems~\ref{ThmI}, \ref{ThmPut} and \ref{ThmIII}, for both ATM and OTM Asian cases. For $K\neq S_0$, the contribution of the diffusive component cancels out in the ratios (\ref{JumpLim}), and the limit is given by the jump coefficients $a_{C,P}$. 
On the other hand, at the $K=S_0$ point the diffusive component is of $O(\sqrt{T})$ and dominates over the jump contribution which is $O(T)$. 

One peculiar feature of this approximation is the discontinuity of option prices across the $K=S_0$ point. This is due to the difference between $a_C(S_0)$ and $a_P(S_0)$. In general these constants are different. For example, in the Merton model with parameters (\ref{MJDparam}) we have, at $K=S_0$, 
\begin{equation}
a_C(S_0)/S_0 = 0.00215\,,\qquad a_P(S_0)/S_0 = 0.0269\,.
\end{equation}

The solution to this apparent puzzle is that the diffusive term dominates at the ATM point, such that the jump term is smaller by a factor of $\sqrt{T}$. Thus this discontinuity becomes parametrically small as $T\to 0$ and vanishes for sufficiently small maturity $T$.

\subsection{Merton jump-diffusion model}

We start by considering the pricing of Asian options in the Merton jump-diffusion model. For the numerical tests we use the same model parameters as Fusai and Meucci \cite{Fusai2008} and Cai et al. \cite{CSK}:
\begin{equation}\label{MJDparam}
\sigma=0.126\,, \quad
\lambda = 0.175\,,\quad
\delta = 0.339\,,\quad
\alpha = -0.39\,.
\end{equation}
The interest rate and dividends rate are $r=q=0$. The spot price
is $S_0=1000$. 

The analytical approximations (\ref{appC}) (for $K\geq S_0$) and (\ref{appP}) (for $K\leq S_0$) for the Asian option prices in the Merton model are shown in 
Figure~\ref{Fig:Merton} as the solid blue curves. The diffusive component is shown as the dashed blue curve. 
The theoretical result is compared against a Monte Carlo (MC) simulation, shown as the red dots with error bars. The MC simulation uses $N_{\mathrm{MC}}=100k$ paths and a time discretization with $N=100$ time steps.
The two panels correspond to options with maturity $T=\frac{1}{52}$ (left) and $T = \frac{1}{12}$ (right).

For short maturity (1 week), the diffusive component is significant only in a small region around the ATM point. Away from this point, the jump contribution dominates the price of the Asian option. In this region the agreement of the MC simulation with the analytical approximation is very good. 

As maturity increases, the diffusive component increases and becomes significant in a wider range of strikes away from the ATM point. For maturities larger than 1 month the analytical approximation underestimates the correct Asian option price obtained from the MC simulation. The reason for this discrepancy may be due to neglected 
$O(T^2)$ contributions which could become important for larger maturity.

The right panel of Figure~\ref{Fig:Merton} shows also the discontinuity of the 
analytical approximation at the ATM point $k=1$, which increases with maturity. 
As discussed above, this is a feature of the analytical approximation adopted here, and disappears for sufficiently small maturity. 

A more detailed comparison is shown in Table~\ref{tab:MJD} which lists the numerical values for several points 
in the left panel of Figure~\ref{Fig:Merton}. 

\begin{figure}[htbp!]
\centering
\includegraphics[scale=0.5]{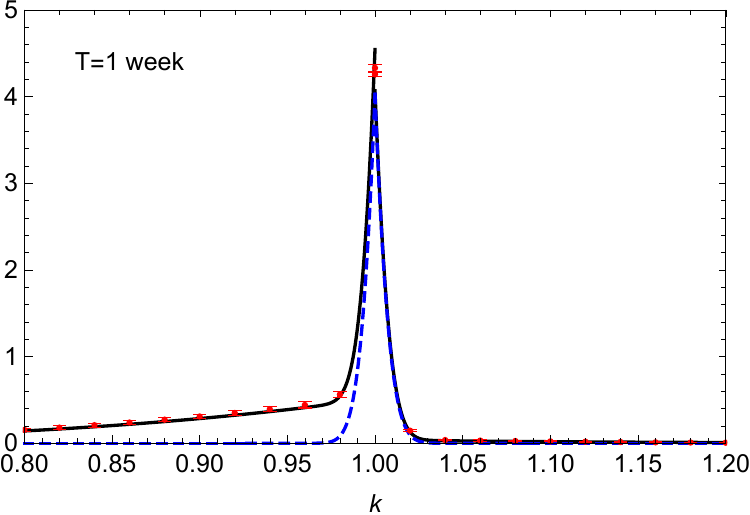}
\includegraphics[scale=0.5]{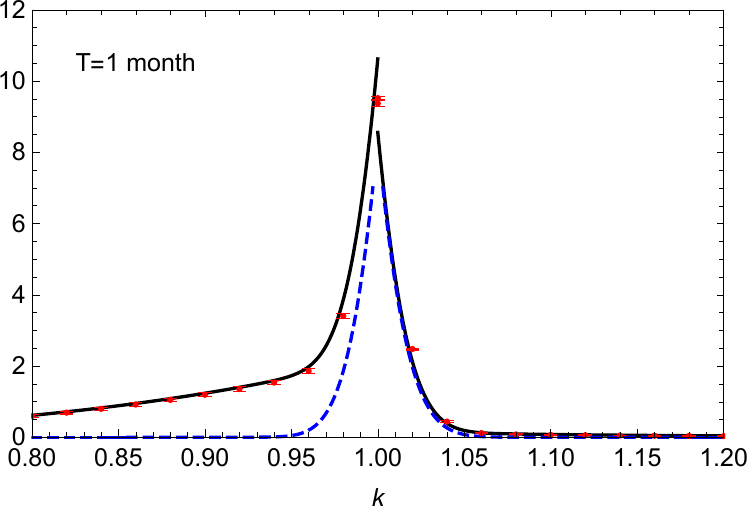}
\includegraphics[scale=0.75]{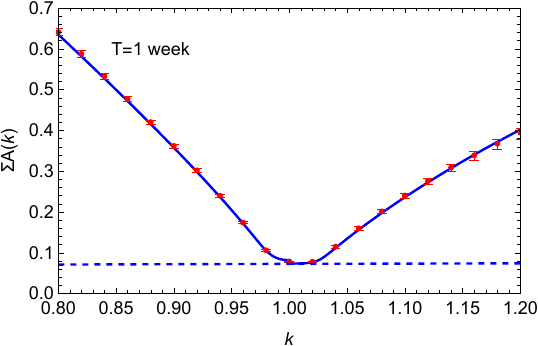}
\includegraphics[scale=0.75]{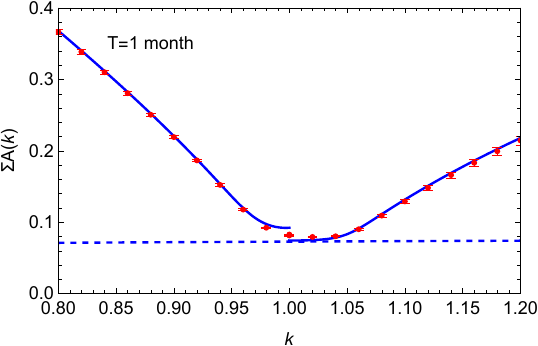}
\caption{Above: Asian option prices in the Merton model for maturity $T=\frac{1}{52}$ (1 week) (left) and $T=\frac{1}{12}$ (1 month) (right). The plots show the Asian put option for $k\leq 1$ and the Asian call options for $k\geq 1$ vs moneyness $k=K/S_0$. The solid black curve shows the analytical approximation (\ref{appC}), (\ref{appP}), and the dashed blue curve shows the diffusive component. The red dots with error bars show the results of a MC simulation described in text. 
Below: The Asian implied volatility $\Sigma_A(K,T)$ for the same maturities. }
\label{Fig:Merton}
\end{figure}

In the lower plots of Figure~\ref{Fig:Merton} we compare also the implied volatilities $\Sigma_A(K,T)$ of the Asian options obtained from the analytical approximations (\ref{appC}), (\ref{appP}) with a MC simulation. The implied volatility of the Asian option $\Sigma_A(K,T)$ is defined as that volatility which reproduces the Asian option price when substituted into the Black-Scholes formula, with forward $A(T)$ defined as in (\ref{Afwddef}). In the absence of jumps, the short-maturity of the Asian option implied volatility becomes
$\lim_{\lambda \to 0} \Sigma_A(K,T) = \Sigma_{\mathrm{LN}}(K/S_0)$, with $\Sigma_{\mathrm{LN}}(K/S_0)$ given above in (\ref{SigLN}). 

The theoretical result for $\Sigma_A(K,T)$ following from the analytical approximation is shown the blue curves. The impact of the jumps is an increase in the Asian implied volatility compared to $\Sigma_{\mathrm{LN}}(K/S_0)$ shown as the dashed blue curve.
The theoretical approximation is discontinuous at the ATM point. This is due to the difference in $a_C(S_0),a_P(S_0)$ discussed above. This discontinuity vanishes as $T\to 0$, when the Asian implied volatility is dominated by the diffusive component.
The analytical approximation is in good agreement with the Monte Carlo simulation (red dots with error bars). 


\begin{table}[htbp]
  \centering
  \caption{Asian call/put option prices with maturity 
$T=1/52$ (1 week) and parameters (\ref{MJDparam}) in the Merton jump diffusion model. 
Columns 2 and 3 show the results of a Monte Carlo simulation.
Columns 4 and 5 show the jump and diffusive components in (\ref{appC}) and (\ref{appP}) and column 6 shows their sum.}
    \begin{tabular}{c|cc|ccc}
    \hline
    $K$ & MC price & Std. Err.  & $a_{C,P}(K) T$ & $C_{BS}(K,T)$ & total \\
    \hline
  960 &  0.4413 & 0.0348 &  $a_P T=0.4112$  & 0.0001 & 0.4112 \\ 
  980 &  0.5623 & 0.0374 &  $a_P T=0.4617$  & 0.0831 & 0.5448 \\ 
 1000 &  4.3289 & 0.0434 &  $a_P T=0.5174$  &  4.0245 &  4.5419 \\
\hline
 1000 & 4.3289  & 0.0434 &  $a_C T=0.0413$  &  4.0245 & 4.0659 \\
 1020 & 0.1449  & 0.0096 &  $a_C T=0.0343$  &  0.0966 & 0.1309 \\
 1040 & 0.0369 & 0.0083 &  $a_C T=0.0289$  &  0.0001 & 0.0290 \\
    \hline
    \end{tabular}%
  \label{tab:MJD}%
\end{table}%


\begin{table}[!ht]
\centering
\caption{Asian options pricing in the double exponential jump model with parameters (\ref{Kouparams}) including also a diffusive component with volatility $\sigma$. $S_0=1000,T=1/52$.}
\begin{tabular}{|ccccc||ccccc|}
\hline
$\sigma$ & $k$ & theory & MC price & Std Err 
    & $\sigma$ & $K$ & theory & MC price & Std Err \\ \hline\hline
 0 & 0.9 & 0.010 & 0.014 & 0.003 
    &  0.3 & 0.9 & 0.010 & 0.015 & 0.003 \\ 
 0 & 0.95 & 0.061 & 0.064 & 0.006 
    & 0.3 & 0.95 & 0.194 & 0.217 & 0.008 \\ 
 0 & 1 & 0.444/0.721 & 0.698 & 0.014 
    & 0.3 & 1 & 10.03/10.30 & 9.782 & 0.046 \\ 
 0 & 1.05 & 0.128 &  0.124 & 0.009 
    & 0.3 & 1.05 & 0.326 & 0.370 & 0.013 \\ 
 0 & 1.1 & 0.032 & 0.041 & 0.007 
    & 0.3 & 1.1 & 0.032 & 0.042 & 0.006 \\ 
 \hline
0.1 & 0.9 & 0.010 & 0.007 & 0.002 
      & 0.4 & 0.9 & 0.014 & 0.023 & 0.003 \\ 
0.1 & 0.95 & 0.061 & 0.062 & 0.005 
      & 0.4 & 0.95 & 0.766 & 0.782 & 0.015  \\ 
0.1 & 1 & 3.64/3.91 & 3.591 & 0.020 
      & 0.4 & 1 & 13.22/13.50 & 13.033 & 0.059 \\ 
0.1 & 1.05 & 0.173 & 0.125 & 0.009 
      & 0.4 & 1.05 & 1.059 & 1.066 & 0.019 \\ 
0.1 & 1.1 & 0.032 & 0.030 & 0.005 
      & 0.4 & 1.1 & 0.047 & 0.064 & 0.006 \\ 
\hline
0.2 & 0.9 & 0.010 & 0.009 & 0.002 
      & 0.5 & 0.9 & 0.056 & 0.062 & 0.004 \\ 
0.2 & 0.95 & 0.064 & 0.078 & 0.007 
      & 0.5 & 0.95 & 1.874 & 1.856 & 0.023 \\ 
0.2 & 1 & 6.83/7.11 & 6.721 & 0.032 
      & 0.5 & 1 & 16.41/16.69 & 16.113 & 0.073 \\ 
0.2 & 1.05 & 0.213 & 0.142 & 0.010 
      & 0.5 & 1.05 & 2.380 & 2.400 & 0.030 \\ 
0.2 & 1.1 & 0.032 & 0.030 & 0.005 
      & 0.5 & 1.1 & 0.162 & 0.185 & 0.009 \\ \hline
\end{tabular}
\label{tab:Kou}%
\end{table}

\subsection{Double Exponential Jump model}

Next we consider the pricing of Asian options under the double
exponential model (\ref{doubleExp}) \cite{Kou}. 
We use similar model  parameters as in \cite{CaiKou}
\begin{equation}\label{Kouparams}
S_0=1000,\quad r=0.0,\quad \alpha=0.6,
\quad
q_1 =1-\alpha= 0.4, \quad\eta_1 = \eta_2 = 25.
\end{equation}
The intensity parameter is chosen as $\lambda=3$.
The parameters $\eta_1,\eta_2$ satisfy the condition $\eta_1,
\eta_2>2$ which is required by Assumption~\ref{Assump1}.

In Table~\ref{tab:Kou} we compare the analytical approximation for Asian option prices (\ref{appC}), (\ref{appP}) 
in the double exponential model with parameters (\ref{Kouparams}) and maturity $T=1/52$
against a Monte Carlo simulation of the model.
The model includes also a diffusive component with constant volatility $\sigma$ which is varied in the range
$\sigma=0.0- 0.5$.
The Monte Carlo simulation used the same parameters as for the Merton model ($N_{\mathrm{MC}}=100k$ paths and
$N=100$ time steps). 

As discussed, the analytical approximation has a two-fold ambiguity
at the ATM point, due to the discontinuity of the $a_C/a_P$ jump coefficients at this
point. The two values are shown in Table~\ref{tab:Kou}  as $P/C$.
For $\sigma=0$ the agreement between the theoretical approximation based on the short-maturity expansion
and the MC simulation is reasonably good for strikes sufficiently far away 
from the ATM point. 
As the volatility $\sigma$ increases, the diffusive component
starts to dominate the ATM Asian option price. Also the discontinuity of the analytical approximation decreases (in relative value), and the agreement with the MC simulation improves.

In Figure~\ref{Fig:Kou} we compare the implied volatilities of the Asian options 
$\Sigma_A(k,T)$ from the analytical approximation with those obtained from the MC simulation.
These plots assume a pure jump model $\sigma=0$, and the Asian option maturity is
$T=1/52$ (left plot) and $T=1/12$ (right plot).
The agreement is reasonably good for both cases, within the MC errors. 

\begin{figure}[t]
\centering
\includegraphics[width=2.7in]{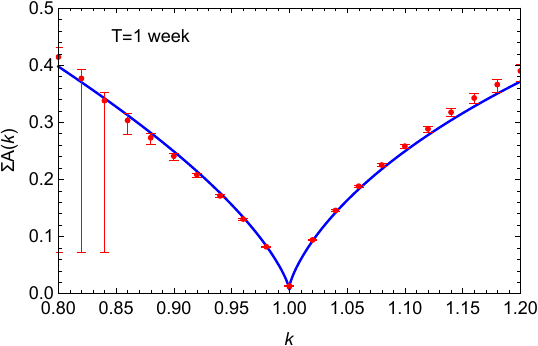}
\includegraphics[width=2.7in]{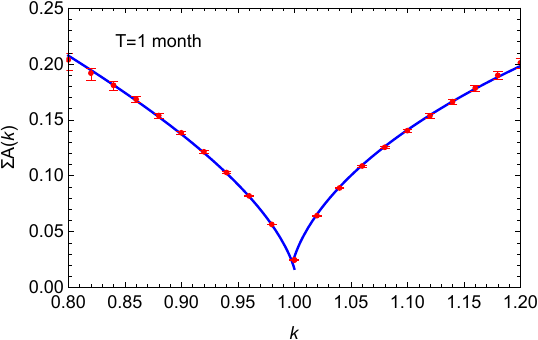}
\caption{
Implied volatility of Asian options $\Sigma_A(k,T)$ in the double exponential jump model with
maturity $T=1/52$ (1 week) (left) and $T=1/12$ (1 month) (right). The solid blue curve shows the 
short-maturity approximation and the red dots 
show the results of a MC simulation as described in text. 
Model parameters are as in (\ref{Kouparams}), and the diffusive volatility is $\sigma=0$.}
\label{Fig:Kou}
\end{figure}

\subsection{Variance Gamma model}

We present in this section test results for Asian options in the Variance
Gamma model. For this test we use the model parameters
\begin{equation}\label{VGparams}
\sigma = 0.4344\,,\quad
\nu = 0.1083\,,\quad
\theta = -0.3726\,,\quad
\eta = 0.0051\,,
\end{equation}
and $r=q=0$. 
These parameters were obtained in \cite{Carr2002} from a nonlinear square estimation from option prices on IBM stock with 1 month maturity.  They calibrated to 
an extended version of the model
of the form $S_t = \hat S_t e^{X_t^{\mathrm{VG}}-\mu t}$ with $\hat S_t = S_0 e^{\eta \hat W_t + (r - q - \frac12\eta^2)t}$ a geometric Brownian motion independent of the VG process, with volatility $\eta$. Expressed in CGMY notations using (\ref{CGMY2VG}) these parameters read
$M=12.062, G=8.113, C=9.234$. It is clear that the condition (\ref{VGcond}) is satisfied by these parameters. 

Table~\ref{tab:VG2} shows the short-maturity asymptotic prediction for the Asian call option in the VG model from Proposition~\ref{prop:VG},  comparing with MC simulations with several option maturities. The MC simulation used 100k paths and was performed in discrete time 
with $N=100$ time steps. The comparison is meaningful only away from the ATM point, as the short maturity result of Proposition~\ref{prop:VG} holds only for OTM options.
The agreement is good for maturities up to 1 week. 

In Figure~\ref{Fig:VG} we compare the implied volatilities of the Asian options from the short maturity asymptotic result of Proposition~\ref{prop:VG} with a Monte Carlo simulation, for maturities $T=1/52$ (left) and $T=1/12$ (right).
We observe again good agreement for the shorter maturity of 1 week, and larger discrepancies for the longer maturity of 1 month.

Extending the applicability of the short maturity expansion to longer maturities requires that the $O(T^2)$ is included as well. 
It is possible that these corrections can be computed using an extension of the method proposed in \cite{FLF} for European 
options in exponential L\'evy models.

\begin{table}[htbp]
  \centering
  \caption{Numerical results for Asian call options in the VG model with parameters (\ref{VGparams}). The second column is the short maturity asymptotic result from Proposition~\ref{prop:VG}, and 
columns 3-5 show $1000\frac{1}{T} C_A(T)$ for several maturities,
obtained by MC simulation as described in text. }
    \begin{tabular}{c|cccc}
    \hline
    $k=K/S_0$ & short maturity & $T=\frac{1}{252}$  & $T=\frac{1}{52}$  
                       & $T=\frac{1}{12}$\\
    \hline\hline
1.00 & 399.55 & $490.9\pm 7.8$ & $410.5 \pm 3.5$ & $286.5\pm 1.7$ \\
1.02 & 166.79 & $162.6\pm 6.2$ & $171.5\pm 2.9$ & $178.6\pm 1.5$ \\
1.04 & 96.93 & $94.8\pm 5.1$ & $103.3 \pm 2.4$ & $118.9\pm 1.3$ \\
1.06 & 61.53 & $59.7\pm 4.3$ & $67.2\pm 2.0$ & $82.2\pm 1.1$ \\
1.08 & 41.06 & $39.9\pm 3.7$ & $45.7\pm 1.7$ & $58.4\pm 1.0$ \\
1.10 & 28.36 & $27.8\pm 3.2$ & $31.9\pm 1.4$ & $42.4\pm 0.8$ \\
1.20 & 6.00 & $6.7\pm 1.7$ & $6.62\pm 0.7$ & $10.3\pm 0.4$ \\
    \hline
    \end{tabular}%
  \label{tab:VG2}%
\end{table}%

\begin{figure}[t]
\centering
\includegraphics[width=2.7in]{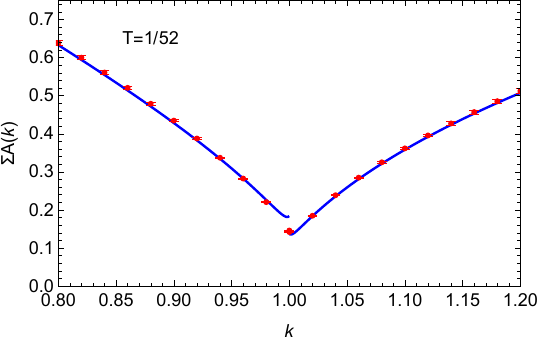}
\includegraphics[width=2.7in]{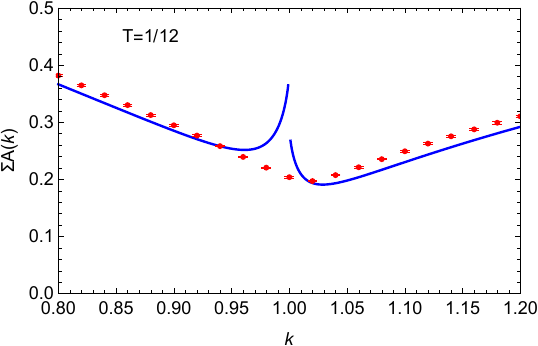}
\caption{
Implied volatility of Asian options in the Variance Gamma model.
The solid blue curves show the leading short-maturity asymptotic result of Proposition~\ref{prop:VG}, and the red dots show the results of the 
MC simulation described in text. 
The Asian option has maturity $T=1/52$ (left) and $T=1/12$ (right).
Model parameters are as in (\ref{VGparams}) and $S_0=1000$.}
\label{Fig:VG}
\end{figure}

\subsection{Numerical tests for floating strike Asian options}

We present in this section numerical tests for the short-maturity asymptotics of floating strike Asian options. For simplicity we restrict these tests to the Merton jump-diffusion
model for which the explicit result for the short-maturity asymptotics was presented in Proposition~\ref{prop:floatMerton}.

The floating strike Asian options are particular cases of the
so-called ``generalized Asian options'' with payoff $(\kappa S_T - A_T - K)^+$ (call) and $(A_T - \kappa S_T - K)^+$ (put), corresponding to $K=0$.
The analytical approximation proposed in this paper (see (\ref{appC}), (\ref{appP}))
is not appropriate for these options, as the underlying asset 
$B_T := \kappa S_T - A_T$ is not positive definite. Thus the Black-Scholes form 
of the diffusive component cannot be used.
A more appropriate treatment would use a Bachelier pricing formula for the 
diffusive component, similar to the approach proposed in Section~4.3 of \cite{AsianForwardStart} for the generalized Asian options in local volatility models.

We present here a direct test of the short-maturity asymptotic result for OTM floating strike Asian options. 
We compare in Table~\ref{tab:floatP} the asymptotic result of Proposition~\ref{prop:floatMerton} (second column) for $1000 \frac{1}{T} P_f(T)$ against
MC simulation results with several maturities from $T=\frac{1}{252}$ (one day)
to $T=\frac{1}{12}$ (1 month). 
For the test we use the same model parameters as in (\ref{MJDparam}). The
parameters of the MC simulation are the same as above. The MC results are
in good agreement with the OTM asymptotic prediction in the second column,
for values of $\kappa$ sufficiently different from 1. 

\begin{table}[htbp]
  \centering
  \caption{Numerical results for floating strike Asian put options in the Merton model with parameters (\ref{MJDparam}). The table shows $1000\frac{1}{T} P_f(\kappa,T)$ for several maturities. The second column gives the short maturity asymptotic result, and the remaining columns show MC simulation results.}
    \begin{tabular}{c|cccc}
    \hline
    $\kappa$ & short maturity & $T=\frac{1}{252}$  & $T=\frac{1}{52}$  
                    &  $T=\frac{1}{12}$ \\
    \hline\hline
    1.0 & 26.904 & $476.50\pm 4.86$ & $223.12\pm 2.19$ & $113.57\pm 1.03$ \\
1.02 & 25.033 & $26.19\pm 4.24$ & $29.43\pm 1.93$  & $43.90\pm 0.95$ \\
1.04 & 23.335 & $24.41 \pm 4.07$ & $23.34\pm 1.86$ & $25.52\pm 0.90$ \\
1.06 & 21.774 & $22.74\pm 3.90$ & $21.75\pm 1.79$ & $21.88\pm 0.87$ \\
1.08 & 20.332 & $21.18\pm 3.74$ & $20.32\pm 1.72$ & $20.35\pm 0.84$ \\
1.10 & 18.995 & $19.75\pm 3.58$ & $18.98\pm 1.66$ & $19.07\pm 0.81$ \\
\hline
1.12 & 17.754 & $18.36\pm 3.43$ & $17.74\pm 1.60$ & $17.87\pm 0.78$ \\
1.14 & 16.598 & $17.05\pm 3.29$ & $16.58\pm 1.54$ & $16.76\pm 0.75$ \\
1.16 & 15.523 & $15.90\pm 3.16$ & $15.49\pm 1.49$ & $15.73\pm 0.72$ \\
1.18 & 14.520 & $14.83\pm 3.02$ & $14.50\pm 1.43$ & $14.76\pm 0.70$ \\
1.20 & 13.584 & $13.77 \pm 2.90$ & $13.60\pm 1.38$ & $13.87\pm 0.67$ \\
    \hline
    \end{tabular}%
  \label{tab:floatP}%
\end{table}%

\section*{Acknowledgements}
Lingjiong Zhu is partially supported by the grants NSF DMS-2053454, NSF DMS-2208303.

\appendix

\section{Technical Proofs}\label{sec:proofs}

\begin{proof}[Proof of Theorem~\ref{ThmI}]
Note that
\begin{align}\label{ThreeTerms}
C(T)&=e^{-rT}\mathbb{E}\left[\left(A(T)-K\right)^{+}1_{N_{T}=0}\right]
+e^{-rT}\mathbb{E}\left[\left(A(T)-K\right)^{+}1_{N_{T}=1}\right]
\\
&\qquad
+e^{-rT}\mathbb{E}\left[\left(A(T)-K\right)^{+}1_{N_{T}\geq 2}\right],
\nonumber
\end{align}
where $A(T)$ is defined in \eqref{A:T:defn}.

Next, let us analyze each of the three terms in \eqref{ThreeTerms} carefully.

Let us first consider the first term in \eqref{ThreeTerms}.
The probability that there is no Poisson jump is given by $\mathbb{P}(N_{T}=0)=e^{-\lambda T}$.
Conditional on zero jumps on the interval $[0,T]$, $S_{t}=\hat{S}_{t}$ process satisfies
the local volatility model:
\begin{equation}
\hat{S}_{t}=S_{0}e^{\int_{0}^{t}\sigma(\hat{S}_{s})dW_{s}
+(r-q-\lambda\mu)t
-\int_{0}^{t}\frac{1}{2}\sigma^{2}(\hat{S}_{s})ds}.
\end{equation}
By using the OTM Asian call option result from Pirjol and Zhu \cite{PZLV}, we get
\begin{align}
e^{-rT}\mathbb{E}\left[\left(A_{T}-K\right)^{+}1_{N_{T}=0}\right]
&=e^{-rT}e^{-\lambda T}\mathbb{E}\left[\left(\hat{A}_{T}-K\right)^{+}\right]
\\
&=e^{-rT}e^{-\lambda T}e^{-\mathcal{I}(K,S_{0})\frac{1}{T}+o(\frac{1}{T})},\nonumber
\end{align}
as $T\rightarrow 0$, where $A(T)$ is defined in \eqref{A:T:defn} and 
\begin{equation}\label{hat:A:T:defn}
\hat{A}_{T}:=\frac{1}{T}\int_{0}^{T}\hat{S}_{t}dt.    
\end{equation}

Next, let us consider the third term in \eqref{ThreeTerms}.
By H\"{o}lder's inequality, for any $\frac{1}{p}+\frac{1}{p'}=1$, where $p,p'>1$, we have
\begin{align}\label{17}
&e^{-rT}\mathbb{E}\left[\left(A_{T}-K\right)^{+}1_{N_{T}\geq 2}\right]
\\
&\leq 
e^{-rT}\mathbb{E}\left[A_{T}1_{N_{T}\geq 2}\right]
+e^{-rT}\mathbb{E}\left[K1_{N_{T}\geq 2}\right]
\nonumber
\\
&\leq e^{-rT}\left(\mathbb{E}\left[A_{T}^{p}\right]\right)^{1/p}
\left(\mathbb{E}[(1_{N_{T}\geq 2})^{p'}]\right)^{1/p'}
+e^{-rT}K\mathbb{P}(N_{T}\geq 2)
\nonumber
\\
&\leq e^{-rT}C\cdot O(T^{2/p'})+e^{-rT}K\cdot O(T^{2}).
\nonumber
\end{align}
Let us choose $1<p'<2$, which is ensured by taking $p$ to be some value greater than $2$, which is made possible by Assumption~\ref{Assump1}. Then, the third term in \eqref{ThreeTerms} is bounded above by $O(T^{2/p'})$, with $2/p' > 1$. 
Thus we proved that the third term in \eqref{ThreeTerms} is $o(T)$.

We give next the proof of the estimate in the last line of equation (\ref{17}).
First, note that
$\mathbb{P}(N_{T}\geq 2)=1-e^{-\lambda T}-\lambda Te^{-\lambda T}=O(T^{2})$
as $T\rightarrow 0$.
Moreover, note that 
$\frac{S_{t}}{e^{(r-q)t}}=:M_{t}$ is a positive martingale. Therefore by the Doob's martingale inequality,
\begin{align}
\left(\mathbb{E}\left[A_{T}^{p}\right]\right)^{1/p}
\leq\left(\mathbb{E}\left[\max_{0\leq t\leq T}S_{t}^{p}\right]\right)^{1/p}
\leq e^{|r-q|T}\left(\mathbb{E}\left[\max_{0\leq t\leq T}M_{t}^{p}\right]\right)^{1/p}
\leq e^{|r-q|T}C_{p}\left(\mathbb{E}\left[M_{T}^{p}\right]\right)^{1/p}\,,
\end{align}
with $C_p = \frac{p}{p-1}$. Furthermore, we have
\begin{align}
\left(\mathbb{E}\left[M_{T}^{p}\right]\right)^{1/p}
&=e^{-(r-q)T}\left(\mathbb{E}\left[S_{T}^{p}\right]\right)^{1/p}
\nonumber
\\
&=e^{-(r-q)T}\left(\mathbb{E}\left[S_{0}^{p}e^{p\int_{0}^{T}\sigma(\hat{S}_{s})dW_{s}
+p(r-q-\lambda\mu)T
-p\int_{0}^{T}\frac{1}{2}\sigma^{2}(\hat{S}_{s})ds
+p\sum_{i=1}^{N_{T}}Y_{i}}\right]\right)^{1/p}
\nonumber
\\
&= e^{-(r-q)T}S_{0}e^{(r-q-\lambda\mu)T}
\nonumber
\\
&\qquad\cdot
\left(\mathbb{E}\left[S_{0}^{p}e^{p\int_{0}^{T}\sigma(\hat{S}_{s})dW_{s}
-\frac{\alpha}{2}p^{2}\int_{0}^{T}\sigma^{2}(\hat{S}_{s})ds+(\frac{\alpha}{2}p^{2}-\frac{1}{2}p)\int_{0}^{T}\sigma^{2}(\hat{S}_{s})ds
+p\sum_{i=1}^{N_{T}}Y_{i}}\right]\right)^{1/p}
\nonumber
\\
&\leq e^{-(r-q)T}S_{0}e^{(r-q-\lambda\mu)T}
\left(\mathbb{E}\left[S_{0}^{p}e^{\alpha p\int_{0}^{T}\sigma(\hat{S}_{s})dW_{s}
-\frac{1}{2}(\alpha p)^{2}\int_{0}^{T}\sigma^{2}(\hat{S}_{s})ds}\right]\right)^{\frac{1}{p\alpha}}
\nonumber
\\
&\qquad\qquad\qquad\cdot
\left(\mathbb{E}\left[e^{(\frac{\beta\alpha}{2}p^{2}-\frac{\beta}{2}p)\int_{0}^{T}\sigma^{2}(\hat{S}_{s})ds
+\beta p\sum_{i=1}^{N_{T}}Y_{i}}\right]\right)^{\frac{1}{p\beta}}
\nonumber
\\
&\leq S_{0}e^{-\lambda\mu T}e^{(\frac{\alpha}{2}p-\frac{1}{2})T(\overline{\sigma})^{2}}
\left(\mathbb{E}\left[e^{\beta p\sum_{i=1}^{N_{T}}Y_{i}}\right]\right)^{\frac{1}{p\beta}}
\nonumber
\\
&= S_{0}e^{-\lambda\mu T}e^{(\frac{\alpha}{2}p-\frac{1}{2})T(\overline{\sigma})^{2}}
e^{\frac{\lambda T}{p\beta}(\mathbb{E}[e^{p\beta Y_{1}}]-1)},
\nonumber
\end{align}
where we have applied H\"{o}lder's inequality with $\alpha,\beta>1$ and $\frac{1}{\alpha}+\frac{1}{\beta}=1$.
This is uniformly bounded as $T\rightarrow 0$, 
since under our Assumption~\ref{Assump1}, we have $\mathbb{E}[e^{p\beta Y_{1}}]<\infty$ since $p>2$ can be chosen to be close to $2$ and $\beta>1$ be chosen to be close to $1$.
In the last two lines, we used the inequality $\frac{1}{T} \int_0^T \sigma^2(\hat S_s) ds\leq(\overline{\sigma})^{2}$ which follows from the technical assumption \eqref{assumpI}.

Finally, let us analyze the second term in \eqref{ThreeTerms}.
Note that $\mathbb{P}(N_{T}=1)=e^{-\lambda T}\lambda T$.
Conditional on $N_{T}=1$, the occurrence time
of the jump is uniformly distributed on $[0,T]$.
Therefore, we have
\begin{align}
&e^{-rT}\mathbb{E}\left[\left(A_{T}-K\right)^{+}1_{N_{T}=1}\right]
\\
&=e^{-rT}e^{-\lambda T}\lambda\int_{-\infty}^{\infty}\int_{0}^{T}
\mathbb{E}\left[\left(\frac{1}{T}\int_{0}^{T}\hat{S}_{s}e^{y1_{s\geq t}}ds-K\right)^{+}\right]
dtdP(y).
\nonumber
\end{align}

By changing the variable in the integration, we get
\begin{align}
&e^{-rT}\mathbb{E}\left[\left(A_{T}-K\right)^{+}1_{N_{T}=1}\right]
\\
&=e^{-rT}e^{-\lambda T}\lambda T\int_{-\infty}^{\infty}\int_{0}^{1}
\mathbb{E}\left[\left(\int_{0}^{1}\hat{S}_{sT}e^{y1_{s\geq t}}ds-K\right)^{+}\right]
dtdP(y).
\nonumber
\end{align}
Next, let us prove that
\begin{align}\label{ClaimI}
&\lim_{T\rightarrow 0}
\int_{-\infty}^{\infty}\int_{0}^{1}
\mathbb{E}\left[\left(\int_{0}^{1}\hat{S}_{sT}e^{y1_{s\geq t}}ds-K\right)^{+}\right]
dtdP(y)
\\
&=\int_{-\infty}^{\infty}\int_{0}^{1}
\left(\int_{0}^{1}S_{0}e^{y1_{s\geq t}}ds-K\right)^{+}dtdP(y).
\nonumber
\end{align}
We can use the put-call parity and write:
\begin{align}\label{TwoTerms2}
&\int_{-\infty}^{\infty}\int_{0}^{1}
\mathbb{E}\left[\left(\int_{0}^{1}\hat{S}_{sT}e^{y1_{s\geq t}}ds-K\right)^{+}\right]
dtdP(y)
\\
&=\int_{-\infty}^{\infty}\int_{0}^{1}
\mathbb{E}\left[\left(K-\int_{0}^{1}\hat{S}_{sT}e^{y1_{s\geq t}}ds\right)^{+}\right]
dtdP(y)
\nonumber
\\
&\qquad\qquad\qquad
+\int_{-\infty}^{\infty}\int_{0}^{1}
\mathbb{E}\left[\left(\int_{0}^{1}\hat{S}_{sT}e^{y1_{s\geq t}}ds-K\right)\right]
dtdP(y).
\nonumber
\end{align}
For any $0<s<1$, $0<t<1$ and $y\in\mathbb{R}$, we have $\hat{S}_{sT}\rightarrow S_{0}$
a.s. as $T\rightarrow 0$ since this process is continuous a.s. Therefore, by bounded convergence theorem,
for the first term in \eqref{TwoTerms2}, we have
\begin{align}
&\lim_{T\rightarrow 0}\int_{-\infty}^{\infty}\int_{0}^{1}
\mathbb{E}\left[\left(K-\int_{0}^{1}\hat{S}_{sT}e^{y1_{s\geq t}}ds\right)^{+}\right]dtdP(y)
\\
&=\int_{-\infty}^{\infty}\int_{0}^{1}
\mathbb{E}\left[\left(K-\int_{0}^{1}S_{0}e^{y1_{s\geq t}}ds\right)^{+}\right]dtdP(y).
\nonumber
\end{align}
For the second term in \eqref{TwoTerms2}, we can compute that
\begin{align}
&\int_{-\infty}^{\infty}\int_{0}^{1}
\mathbb{E}\left[\left(\int_{0}^{1}\hat{S}_{sT}e^{y1_{s\geq t}}ds-K\right)\right]dtdP(y)
\\
&=\int_{-\infty}^{\infty}\int_{0}^{1}
\left[\left(\int_{0}^{1}\mathbb{E}[\hat{S}_{sT}]e^{y1_{s\geq t}}ds-K\right)\right]dtdP(y).
\nonumber
\end{align}
For any $0<t<1$, by the definition of $\hat{S}_{sT}$,
we can compute that for any $0<s<1$, 
$\mathbb{E}[\hat{S}_{sT}]=S_{0}e^{(r-q-\lambda\mu)sT}$.
Hence,
\begin{align}
&\int_{-\infty}^{\infty}\int_{0}^{1}
\mathbb{E}\left[\left(\int_{0}^{1}\hat{S}_{sT}e^{y1_{s\geq t}}ds-K\right)\right]dtdP(y)
\\
&=\int_{-\infty}^{\infty}\int_{0}^{1}
\left[\left(\int_{0}^{1}S_{0}e^{(r-q-\lambda\mu)sT}e^{y1_{s\geq t}}ds-K\right)\right]dtdP(y).
\nonumber
\end{align}
It is clear that
\begin{align}
e^{-|r-q-\lambda\mu|T}\int_{-\infty}^{\infty}\int_{0}^{1}\left[\int_{0}^{1}S_{0}e^{y1_{s\geq t}}ds\right]dtdP(y)
&\leq\int_{-\infty}^{\infty}\int_{0}^{1}\left[\int_{0}^{1}S_{0}e^{(r-q-\lambda\mu)sT}e^{y1_{s\geq t}}ds\right]dtdP(y)
\\
&\leq e^{|r-q-\lambda\mu|T}\int_{-\infty}^{\infty}\int_{0}^{1}\left[\int_{0}^{1}S_{0}e^{y1_{s\geq t}}ds\right]dtdP(y).
\nonumber
\end{align}
By letting $T\rightarrow 0$, for the second term in \eqref{TwoTerms2}, we get 
\begin{align}
&\lim_{T\rightarrow 0}
\int_{-\infty}^{\infty}\int_{0}^{1}
\mathbb{E}\left[\left(\int_{0}^{1}\hat{S}_{sT}e^{y1_{s\geq t}}ds-K\right)\right]dtdP(y)
\\
&=\int_{-\infty}^{\infty}\int_{0}^{1}
\left(\int_{0}^{1}S_{0}e^{y1_{s\geq t}}ds-K\right)dtdP(y).
\nonumber
\end{align}
Hence, applying put-call parity again, the claim in \eqref{ClaimI} is proved.

Therefore, we have
\begin{align}
&\lim_{T\rightarrow 0}\frac{1}{T}
e^{-rT}\mathbb{E}\left[\left(A_{T}-K\right)^{+}1_{N_{T}=1}\right]
\\
&=\lambda\int_{-\infty}^{\infty}\int_{0}^{1}
\left(\int_{0}^{1}S_{0}e^{y1_{s\geq t}}ds-K\right)^{+}dtdP(y)
\nonumber
\\
&=\lambda\int_{-\infty}^{\infty}\int_{0}^{1}
\left(S_{0}t+S_{0}e^{y}(1-t)-K\right)^{+}dtdP(y)
\nonumber
\\
&=\lambda\int_{0}^{1}\int_{\log(\frac{K-S_{0}t}{S_{0}(1-t)})}^{\infty}\left(S_{0}t+S_{0}e^{y}(1-t)-K\right)dP(y)dt.
\nonumber
\end{align}
This completes the proof.
\end{proof}

\begin{proof}[Proof of Theorem~\ref{ThmPut}]
The proof of Theorem~\ref{ThmPut} is very similar to the proof of Theorem~\ref{ThmI}
is hence omitted here.
\end{proof}


\begin{proof}[Proof of Theorem~\ref{ThmIII}]
The price of the ATM Asian call option can be decomposed as:
\begin{equation}\label{TwoTerms}
C(T)=e^{-rT}\mathbb{E}\left[\left(A_{T}-S_{0}\right)^{+}1_{N_{T}=0}\right]
+e^{-rT}\mathbb{E}\left[\left(A_{T}-S_{0}\right)^{+}1_{N_{T}\geq 1}\right],
\end{equation}
where $A_{T}$ is defined in \eqref{A:T:defn}.
Conditional on $N_{T}=0$, $S_{t}=\hat{S}_{t}$, where $\hat{S}_{t}$
satisfies the local volatility model without jumps:
\begin{equation}
\frac{d\hat{S}_{t}}{\hat{S}_{t-}}=(r-q)dt
-\lambda\mu dt+\sigma(\hat{S}_{t})dW_{t},
\qquad
\hat{S}_{0}=S_{0},
\end{equation}
and therefore the first term in \eqref{TwoTerms} can be written as
\begin{equation}
e^{-rT}\mathbb{E}\left[\left(A_{T}-S_{0}\right)^{+}1_{N_{T}=0}\right]
=e^{-rT}e^{-\lambda T}\mathbb{E}\left[\left(\hat{A}_{T}-S_{0}\right)^{+}\right] \,,
\end{equation}
where $\hat{A}_{T}$ is defined in \eqref{hat:A:T:defn}.
By the result for ATM Asian options from \cite{PZLV}, we have
\begin{equation}
\lim_{T\rightarrow 0}\frac{1}{\sqrt{T}}\mathbb{E}\left[\left(\hat{A}_{T}-S_{0}\right)^{+}\right]
=\frac{1}{\sqrt{6\pi}}\sigma(S_{0})S_{0}.
\end{equation}
On the other hand, by H\"{o}lder's inequality, for any $p,p'>1$
and $\frac{1}{p}+\frac{1}{p'}=1$, we have
\begin{align}
e^{-rT}\mathbb{E}\left[\left(A_{T}-S_{0}\right)^{+}1_{N_{T}\geq 1}\right]
\leq e^{-rT}\left(\mathbb{E}\left[\left|A_{T}-S_{0}\right|^{p}\right]\right)^{1/p}
\left(\mathbb{E}\left[1_{N_{T}\geq 1}\right]\right)^{1/p'}
\leq 
C\cdot O(T^{1/p'}).
\end{align}
We can choose $1<p'<2$ (which can be achieved by choosing some $p>2$
which is made possible by Assumption~\ref{Assump1}), which implies that this is $o(T^{1/2})$. 
Thus, we get the desired result.
The proof for the asymptotics for ATM Asian put options is very similar
and is hence omitted here.
\end{proof}


\begin{proof}[Proof of Theorem~\ref{Thm:LevyC}]
We can also decompose $X_{t}$ as
\begin{equation}
X_{t}=b_{\epsilon}t+\int_{0}^{t}\int_{|x|\leq\epsilon}x(\mu-\bar{\mu})(dx,ds)
+\int_{0}^{t}\int_{|x|>\epsilon}x\mu(dx,ds) 
= b_{\epsilon} t + M_t^\epsilon + J_{t}^\epsilon\,,
\end{equation}
where we defined
\begin{equation}
J_{t}^{\epsilon}:=\int_{0}^{t}\int_{|x|>\epsilon}x\mu(dx,ds)
\end{equation}
which is a compound Poisson process with jump intensity $\lambda^{\epsilon}=\nu(\{|x|\geq\epsilon\})$
and the jump size distribution $\frac{1}{\lambda^{\epsilon}}1_{|x|\geq\epsilon}\nu(dx)$.
Moreover,
\begin{equation}
M_{t}^{\epsilon} := \int_{0}^{t}\int_{|x|\leq\epsilon}x(\mu-\bar{\mu})(dx,ds)
\end{equation}
is a martingale satisfying $M_0^\epsilon=0$ and with
predictable quadratic variation $\langle M^\epsilon \rangle_t
= t \int_{|x| \leq \epsilon} x^2 \nu(dx)$, 
which is finite due to Assumption~\ref{assump2}.

From the Burkholder-Davis-Gundy inequality, we get for any $w\geq 1$,
\begin{align}
\mathbb{P}\left(\max_{0\leq t\leq T}|M_{t}^{\epsilon}|\geq\delta\right)
&\leq\frac{\mathbb{E}[\max_{0\leq t\leq T}|M_{t}^{\epsilon}|^{w}]}{\delta^{w}}
\\
&\leq\frac{C_{w}\mathbb{E}[\langle M^{\epsilon}\rangle_{T}^{w/2}]}{\delta^{w}}
=\frac{C_{w}}{\delta^{w}}\left(\int_{|x|\leq\epsilon}x^{2}\nu(dx)\right)^{w/2}T^{w/2}\,,
\nonumber
\end{align}
with $C_w>0$ a positive constant and $\int_{|x|\leq\epsilon}x^{2}\nu(dx)$
is finite due to Assumption~\ref{assump2}.

We can estimate the Asian call option price as:
\begin{align}
C(T)&=e^{-rT}\mathbb{E}\left[\left(\frac{1}{T}\int_{0}^{T}
\hat{S}_{t}e^{X_{t}}dt-K\right)^{+}\right]
\\
&=e^{-rT}\mathbb{E}\left[\left(\frac{1}{T}\int_{0}^{T}
\hat{S}_{t}e^{b_{\epsilon}t+M_{t}^{\epsilon}+J_{t}^{\epsilon}}dt-K\right)^{+}1_{\max_{0\leq t\leq T}|M_{t}^{\epsilon}|\geq\delta}\right]
\nonumber
\\
&\qquad
+e^{-rT}\mathbb{E}\left[\left(\frac{1}{T}\int_{0}^{T}
\hat{S}_{t}e^{b_{\epsilon}t+M_{t}^{\epsilon}+J_{t}^{\epsilon}}dt-K\right)^{+}1_{\max_{0\leq t\leq T}|M_{t}^{\epsilon}|\leq\delta}\right]
\nonumber
\\
&\leq e^{-rT}\mathbb{E}\left[\left(\frac{1}{T}\int_{0}^{T}
\hat{S}_{t}e^{b_{\epsilon}t+M_{t}^{\epsilon}+J_{t}^{\epsilon}}dt+K\right)1_{\max_{0\leq t\leq T}|M_{t}^{\epsilon}|\geq\delta}\right]
\nonumber
\\
&\qquad
+e^{-rT}\mathbb{E}\left[\left(\frac{1}{T}\int_{0}^{T}
\hat{S}_{t}e^{b_{\epsilon}t+\delta+J_{t}^{\epsilon}}dt-K\right)^{+}\right]\,,
\nonumber
\end{align}
where we used $e^x \leq e^\delta$ for all $|x|\leq \delta$ in the second term.

Applying the H\"older inequality with $p,p'>1$ and $\frac{1}{p}+\frac{1}{p'}=1$, 
we have further
\begin{align}
C(T) &\leq e^{-rT}\left[\left(\mathbb{E}\left[\left(\frac{1}{T}\int_{0}^{T}
\hat{S}_{t}e^{b_{\epsilon}t+M_{t}^{\epsilon}+J_{t}^{\epsilon}}dt\right)^{p}\right]\right)^{1/p}+K\right]
\left(\mathbb{P}\left(\max_{0\leq t\leq T}|M_{t}^{\epsilon}
|\geq  \delta \right)\right)^{1/p'}
\nonumber
\\
&\qquad
+e^{-rT}\mathbb{E}\left[\left(\frac{1}{T}\int_{0}^{T}
\hat{S}_{t}e^{b_{\epsilon}t+\delta+J_{t}^{\epsilon}}dt-K\right)^{+}\right]
\nonumber
\\
&\leq O\left(T^{\frac{w}{2p'}}\right)+e^{-rT}\mathbb{E}\left[\left(\frac{1}{T}\int_{0}^{T}
\hat{S}_{t}e^{b_{\epsilon}t+\delta+J_{t}^{\epsilon}}dt-K\right)^{+}\right],
\nonumber
\end{align}
where the first term is $o(T)$ by choosing $w>2p'$
and the term $\mathbb{E}\left[\left(\frac{1}{T}\int_{0}^{T}
\hat{S}_{t}e^{b_{\epsilon}t+M_{t}^{\epsilon}+J_{t}^{\epsilon}}dt\right)^{p}\right]$
can be bounded similarly as in the proof of Theorem~\ref{ThmI} 
by applying Assumption~\ref{assump2}.

Hence, we get for sufficiently small $\delta>0$,
\begin{align}\label{CUpperBound}
\limsup_{T\rightarrow 0}\frac{C(T)}{T}
&\leq
\limsup_{T\rightarrow 0}\frac{e^{-rT}}{T}\mathbb{E}\left[\left(\frac{1}{T}\int_{0}^{T}
\hat{S}_{t}e^{b_{\epsilon}t+\delta+J_{t}^{\epsilon}}dt-K\right)^{+}\right]
\\
&\leq
\int_{0}^{1}\int_{\log(\frac{K-S_{0}e^{\delta}t}{S_{0}e^{\delta}(1-t)})}^{\infty}1_{|y|\geq\epsilon}
\left(S_{0}e^{\delta}t+S_{0}e^{\delta}e^{y}(1-t)-K\right)\nu(dy)dt
\nonumber
\\
&=\int_{0}^{1}\int_{\log(\frac{K-S_{0}e^{\delta}t}{S_{0}e^{\delta}(1-t)})}^{\infty}
\left(S_{0}e^{\delta}t+S_{0}e^{\delta}e^{y}(1-t)-K\right)\nu(dy)dt,
\nonumber
\end{align}
where we applied Theorem~\ref{ThmI} (by using Assumption~\ref{assump2}) to obtain the second line above 
and in the last line we took $\epsilon>0$ to be sufficiently small so that $\epsilon<\inf_{0<t<1}\log(\frac{K-S_{0}t}{S_{0}e^{-\delta}(1-t)})=\log(\frac{K}{S_{0}e^{-\delta}})$.

On the other hand, using $e^x \geq e^{-\delta}$ for all $|x|\leq \delta$, we have the lower bound
\begin{align}
C(T)&
\geq 
e^{-rT}\mathbb{E}\left[\left(\frac{1}{T}\int_{0}^{T}
\hat{S}_{t}e^{b_{\epsilon}t+M_{t}^{\epsilon}+J_{t}^{\epsilon}}dt-K\right)^{+}1_{\max_{0\leq t\leq T}|M_{t}^{\epsilon}|\leq\delta}\right]
\\
&\geq
e^{-rT}\mathbb{E}\left[\left(\frac{1}{T}\int_{0}^{T}
\hat{S}_{t}e^{b_{\epsilon}t-\delta+J_{t}^{\epsilon}}dt-K\right)^{+}1_{\max_{0\leq t\leq T}|M_{t}^{\epsilon}|\leq\delta}\right]
\nonumber
\\
&=e^{-rT}\mathbb{E}\left[\left(\frac{1}{T}\int_{0}^{T}
\hat{S}_{t}e^{b_{\epsilon}t-\delta+J_{t}^{\epsilon}}dt-K\right)^{+}\right]
\nonumber
\\
&\qquad\qquad
-e^{-rT}\mathbb{E}\left[\left(\frac{1}{T}\int_{0}^{T}
\hat{S}_{t}e^{b_{\epsilon}t-\delta+J_{t}^{\epsilon}}dt-K\right)^{+}1_{\max_{0\leq t\leq T}|M_{t}^{\epsilon}|\geq\delta}\right].
\nonumber
\end{align}

Proceeding in a similar way as before, we get for sufficiently small $\delta>0$
\begin{align}\label{CLowerBound}
\liminf_{T\rightarrow 0}\frac{C(T)}{T}
&\geq
\int_{0}^{1}\int_{\log(\frac{K-S_{0}e^{-\delta t}}{S_{0}e^{-\delta}(1-t)})}^{\infty}1_{|y|\geq\epsilon}\left(S_{0}e^{-\delta}t+S_{0}e^{-\delta}e^{y}(1-t)-K\right)\nu(dy)dt
\\
&=\int_{0}^{1}\int_{\log(\frac{K-S_{0}t}{S_{0}e^{-\delta}(1-t)})}^{\infty}\left(S_{0}t+S_{0}e^{-\delta}e^{y}(1-t)-K\right)\nu(dy)dt,
\nonumber
\end{align}
where in the last line we take $\epsilon>0$ to be sufficiently small so that $\epsilon<\inf_{0<t<1}\log(\frac{K-S_{0}t}{S_{0}e^{-\delta}(1-t)})=\log(\frac{K}{S_{0}e^{-\delta}})$.

Noting that the choice of $\delta>0$ is arbitrary, we let $\delta\rightarrow 0$.
Combining the lower bound (\ref{CLowerBound}) with the  upper bound (\ref{CUpperBound}), the stated result follows. 
\end{proof}


\begin{proof}[Proof of Theorem~\ref{Thm:LevyP}]
The proof of Theorem~\ref{Thm:LevyP} is very similar to 
the proof of Theorem~\ref{Thm:LevyC} and is hence omitted here.
\end{proof}


\begin{proof}[Proof of Theorem~\ref{Thm:LevyC:floating}]
The proof is similar to that of Theorem~\ref{ThmI} and Theorem~\ref{Thm:LevyC},
and we will only provide an outline here. As in the proof of Theorem~\ref{Thm:LevyC}, 
one can decompose the L\'{e}vy jumps into big jumps and small jumps, 
and it suffices to consider the big jumps, which boils the problem down 
to the local volatility model with compound Poisson jumps that is considered in Theorem~\ref{ThmI}. 
For the diffusion part (local volatility part), the probability it deviates away from $S_{0}$
is exponentially small in $T$, i.e. $e^{-O(1/T)}$ which is negligible. 
The leading order asymptotics of $C_{f}(T)$ as $T\rightarrow 0$, is contributed
by one single jump, and by following the same arguments as in the proof of Theorem~\ref{ThmI}, we obtain
\begin{align}
\lim_{T\rightarrow 0}\frac{C_{f}(T)}{T}
&=\int_{0}^{1}\int_{-\infty}^{\infty}\left(\kappa S_{0}e^{y}-S_{0}t-S_{0}e^{y}(1-t)\right)^{+}\nu(dy)dt
\\
&=\int_{1-\kappa}^{1}\int_{\log(\frac{t}{\kappa-(1-t)})}^{\infty}S_{0}\left(\kappa e^{y}-t-e^{y}(1-t)\right)\nu(dy)dt\,.
\nonumber
\end{align}
This completes the proof.
\end{proof}

\begin{proof}[Proof of Theorem~\ref{Thm:LevyP:floating}]
The proof of Theorem~\ref{Thm:LevyP:floating} is very similar to 
the proof of Theorem~\ref{Thm:LevyC:floating} and is hence omitted here.
\end{proof}


\begin{proof}[Proof of Theorem~\ref{Thm:ATM:floating}]
The proof is similar to that of Theorem~\ref{ThmIII} 
by leveraging the similar estimates as in the proof of Theorem~\ref{ThmIII}
and the short-maturity asymptotics for ATM Asian floating strike options
under the local volatility model (without jumps) in \cite{PZLV}.
\end{proof}

\begin{proof}[Proof of Proposition~\ref{prop:Merton}]
(i) The result for the OTM Asian call option follows from Theorem \ref{ThmI}. 
The $y$ integral appearing in the short maturity asymptotics (\ref{aC}) 
can be evaluated in closed form using the expression (\ref{dPMerton}) for the jump size distribution $dP(y)$ in the Merton model. This yields the result stated. 

(ii) The result for the OTM Asian put options follows from Theorem~\ref{ThmPut} and is analogous to that for the calls, so it is omitted.
\end{proof}

\begin{proof}[Proof of Proposition~\ref{prop:floatMerton}]
The proof of Proposition~\ref{prop:floatMerton} is similar to that of Proposition~\ref{prop:Merton}
and is hence omitted here.
\end{proof}


\begin{proof}[Proof of Proposition~\ref{prop:Kou}]
For OTM Asian call options, i.e., $K>S_{0}$, we have
for any $0<t<1$, $\log(\frac{K-S_{0}t}{S_{0}(1-t)})>0$, and thus we can compute that
\begin{align}
\lim_{T\rightarrow 0}\frac{C(T)}{T}
&=\lambda\int_{0}^{1}\int_{\log(\frac{K-S_{0}t}{S_{0}(1-t)})}^{\infty}\left(S_{0}e^{y}(1-t)-(K-S_{0}t)\right)
\alpha\eta_{1}e^{-\eta_{1}y}dydt
\\
&=\frac{\lambda\alpha}{\eta_{1}-1}S_{0}\int_{0}^{1}\left(\frac{K}{S_{0}}-t\right)^{1-\eta_{1}}(1-t)^{\eta_{1}}dt
\nonumber
\\
&=\frac{\lambda\alpha S_{0}}{\eta_{1}^{2}-1}\left(\frac{K}{S_{0}}\right)^{1-\eta_{1}}
{}_{2}F_{1}\left(1,\eta_{1}-1;\eta_{1}+2;\frac{S_{0}}{K}\right).
\nonumber
\end{align}

(ii) 
The OTM Asian put option result follows from Theorem \ref{ThmPut}.
For this case we have $K<S_{0}$ which implies
$\log(\frac{K-S_{0}t}{S_{0}(1-t)})\leq 0$ for any $0<t<K/S_0$.
Therefore, we can compute that 
\begin{align}
\lim_{T\rightarrow 0}\frac{P(T)}{T}
&=\lambda\int_{0}^{\frac{K}{S_{0}}}\int_{-\infty}^{\log(\frac{K-S_{0}t}{S_{0}(1-t)})}\left((K-S_{0}t)-S_{0}e^{y}(1-t)\right)
(1-\alpha)\eta_{2}e^{\eta_{2}y}dydt
\\
&=\frac{\lambda(1-\alpha)}{\eta_{2}+1}S_{0}\int_{0}^{\frac{K}{S_{0}}}\left(\frac{K}{S_{0}}-t\right)^{1+\eta_{2}}(1-t)^{-\eta_{2}}dt
\nonumber
\\
&=\frac{\lambda(1-\alpha)}{\eta_{2}+1}\frac{K^{2}}{S_{0}}\int_{0}^{1}\left(\frac{S_{0}}{K}-t\right)^{-\eta_{2}}(1-t)^{1+\eta_{2}}dt
\nonumber
\\
&=\frac{\lambda(1-\alpha)S_0}{(\eta_{2}+1)(\eta_{2}+1)}\left(\frac{K}{S_{0}}\right)^{\eta_{2}+2}
{}_{2}F_{1}\left(1,\eta_{2};\eta_{2}+3;\frac{K}{S_{0}}\right).
\nonumber
\end{align}
This completes the proof.
\end{proof} 


\begin{proof}[Proof of Lemma~\ref{lem:VG}]
Substituting the CGMY density (\ref{CGMYmeasure}), the integral appearing in the statement of Assumption~\ref{assump2} takes the form
\begin{align}
I &= \int_{-\infty}^\infty \min(1,x^2) e^{\theta x} \nu(dx) 
\\
&=
C\int_{-\infty}^{-1} e^{(G+\theta )x} \frac{dx}{|x|^{1+Y}} + 
C\int_{-1}^0 |x|^{1-Y} e^{(G+\theta ) x} dx\nonumber
\\
&\qquad\qquad\qquad\qquad +
C \int_0^1 x^{1-Y} e^{-(M-\theta) x} dx + C \int_1^\infty e^{-(M-\theta) x} \frac{dx}{x^{1+Y}} \,.
\nonumber
\end{align}
This is finite for some $\theta>2$ provided that $Y<2$, $G>-2$ and $M>2$. The second condition is automatic since $G>0$. We conclude that Assumption~\ref{assump2} holds if $M>2$ and $Y<2$. 

Expressing $M=1/\eta_p$ and using (\ref{etap}) gives the condition (\ref{VGcond}) for the limiting case of the VG model.
\end{proof}


\begin{proof}[Proof of Proposition~\ref{prop:VG}]
We present only the OTM Asian call case, the Asian put is analogous. 
The starting point is Theorem~\ref{Thm:LevyC}. Substituting the L\'evy measure (\ref{VGmeasure}) in CGMY form with $Y=0$, the $y$-integration has the form:
\begin{align}
&\int_{\log y_0(t)}^\infty (S_0 t + S_0 e^y (1-t)-K) C e^{-M y} \frac{dy}{y}  \\
&= (S_0 t - K) C \int_{\log y_0}^\infty e^{-M y} \frac{dy}{y} +
S_0 (1-t) C \int_{\log y_0}^\infty e^{-(M-1) y} \frac{dy}{y}\,, \nonumber
\end{align}
where denoted here $y_0(t) := \frac{K/S_0-t}{1-t}$. 

The two integrals are similar, and are evaluated in terms of $\mathrm{li}(z)$ by a few changes of variable, as follows. In a first step we introduce $z=e^y$ and obtain
\begin{equation}
J(t,M) = \int_{\log y_0}^\infty e^{-My} \frac{dy}{y}  = 
\int_{y_0}^\infty z^{-M-1} \frac{dz}{\log z} \,.
\end{equation}
Next introduce $dw =  z^{-M-1} dz$ which gives $w = - \frac{1}{M}z^{-M}$.
The integral becomes
\begin{equation}
J(t,M) = \int_0^{\frac{1}{M} (y_0)^{-M}} \frac{dw}{\frac{1}{M} \log(Mw)} = 
\int_0^{(y_0)^{-M}} \frac{du}{\log u}\,.
\end{equation}
Using the definition of the $\mathrm{li}(z)$ function this gives (\ref{Jres}).
\end{proof}


\begin{thebibliography}{99}

\bibitem{Alos2007}
Al\`{o}s, E., L\'{e}on, J. and Vives, J. (2007).
On the short-time behavior of the implied volatility for jump-diffusion models with stochastic volatility.
\textit{Finance and Stochastics}.
\textbf{11}, 571-589.

\bibitem{Alos2022}
Al\`{o}s, E., Nualart, E. and Pravosud, M. (2022).
On the implied volatility of Asian options under stochastic volatility.
arXiv: 2208.01353[q-fin]

\bibitem{Alziary}
Alziary, B., Decamps, J. P. and Koehl, P. F. (1997).
A PDE approach to Asian options: Analytical and Numerical evidence.
\textit{Journal of Banking and Finance}.
\textbf{21}, 613-640.

\bibitem{Andersen2000}
Andersen, L. and Andreasen, J. (2000).
Jump diffusion process: Volatility smile fitting and numerical methods for option pricing.
\textit{Review of Derivatives Research}.
\textbf{4}, 231-262.

\bibitem{Arguin2018}
Arguin, L.P., N.-L. Lin and Wang, T.-H. (2018).
Most-likely path in Asian option pricing under local volatility models.
\textit{International Journal of Theoretical and Applied Finance},
\textbf{21}(5), 1850029.

\bibitem{BayXing} 
Bayraktar, E. and Xing, H. (2011).
Pricing Asian options for jump diffusion.
\textit{Mathematical Finance}. 
{\bf 21}, 117-143.

\bibitem{BenGobMiri}
Benhamou, E., Gobet, E. and Miri, M. (2009).
Smart expansion and fast calibration for jump diffusions.
\textit{Finance and Stochastics}. 
{\bf 13}, 563-589.

\bibitem{Benhamou}
Benhamou, E. (2002).
Fast Fourier transform for discrete Asian options. 
\textit{Journal of Computational Finance}. 
\textbf{6}, 49-61.

\bibitem{BL02}
Boyarchenko, S.I. and Levendorskii, S. (2002).
Non-Gaussian Merton-Black-Scholes Theory.
\textit{Advanced Series on Statistical Science and Applied Probability}. 
Vol.9, World Scientific Publishing Co., 
River Edge NJ, 2002.

\bibitem{CaiKou} 
Cai, N. and Kou, S.~G. (2012).
Pricing Asian options under a hyper-exponential jump diffusion model.
\textit{Operations Research}. 
{\bf 60}(1), 64-77.

\bibitem{CaiLiShi}
Cai, N., Li, C. and Shi, C. (2014). 
Closed-form expansions of discretely-monitored Asian options in diffusion models.
\textit{Mathematics of Operations Research}. 
{\bf 39}(3), 789-822.

\bibitem{CSK} 
Cai, N., Song, Y. and Kou, S.~G.. (2015).
A general framework for pricing Asian options under Markov processes.
\textit{Operations Research}. 
{\bf 63}(3), 540-554.

\bibitem{Cui2018} 
Cui, Z., Lee, C. and Liu, Y. (2018).
Single-transform formulas for pricing Asian options in a general approximation framework under Markov processes.
\textit{European Journal of Operational Research}.
\textbf{266}(3), 1134-1139.

\bibitem{CarvClewlow} 
Carverhill, P. and Clewlow, M. (1990).
Flexible convolution.
\textit{Risk}. 
{\bf 3}(4), 25-29.

\bibitem{Carr2002}
Carr, P., Geman, H., Madan, D. and Yor, M. (2002).
The fine structure of asset returns: An empirical investigation.
\textit{Journal of Business}. 
\textbf{75}, 303-325.

\bibitem{CS} 
Carr, P. and Schr\"oder, M. (2003). 
Bessel processes, the integral of geometric Brownian motion, and Asian options.
\textit{Theory of Probability and its Applications}. 
{\bf 48}, 400-425.

\bibitem{CarrWu} 
Carr, P. and Wu, L. (2003).
What type of process underlies options? A simple robust test.
\textit{The Journal of Finance}. 
{\bf 58}(6), 2581-2610.

\bibitem{ChowLin}
Chow, C.~S. and Lin, H.~J. (2006).
Asian options with jumps.
\textit{Statistics and Probability Letters}.
\textbf{76}, 1983-1993.


\bibitem{Dembo} 
Dembo, A. and Zeitouni, O. 
\textit{Large Deviations Techniques and Applications}. 2nd Edition, Springer, New York, 1998.


\bibitem{DufresneReview} 
Dufresne, D. (2005).
Bessel processes and a functional
of Brownian motion, in M.~Michele and H.~Ben-Ameur (Ed.), 
{\em Numerical Methods in Finance}, 35-57, Springer, 2005.

\bibitem{FL08}
Figueroa-L\'{o}pez, J.~E. (2008).
Small-time moment asymptotics for L\'{e}vy models.
\textit{Statistics and Probability Letters}.
\textbf{78}, 3355-3365.

\bibitem{FLF}
Figueroa-L\'{o}pez, J.~E. and Forde, M. (2012).
The small-maturity smile for exponential L\'{e}vy models.
\textit{SIAM Journal on Financial Mathematics}.
\textbf{3}, 33-65.

\bibitem{FLO}
Figueroa-L\'{o}pez, J.~E. and \'{O}laffson, S. (2016).
Short-time asymptotics for the implied volatility skew under a stochastic volatility model
with L\'{e}vy jumps.
\textit{Finance and Stochastics}.
\textbf{20}, 973-1020.

\bibitem{FPP2013}
Foschi, P., Pagliarani, S. and Pascucci, A. (2013).
Approximations for Asian options in local volatility models.
\textit{Journal of Computational and Applied Mathematics}.
\textbf{237}, 442-459.

\bibitem{FMW}
Fu, F., Madan, D. and Wang, T. (1998).
Pricing continuous time Asian options: a comparison of Monte Carlo and Laplace
transform inversion methods.
\textit{Journal of Computational Finance}.
\textbf{2}, 49-74.

\bibitem{Fusai2004}
Fusai, G. (2004).
Pricing Asian options via Fourier and Laplace transforms.
\textit{Journal of Computational Finance}.
\textbf{7}(3), 87-106.

\bibitem{Fusai2008}
Fusai, G. and Meucci, A. (2008). 
Pricing discretely monitored Asian options under L\'evy processes,
\textit{Journal of Banking and Finance}. 
\textbf{32}(10):2076-2088.

\bibitem{GY}
Geman, H. and Yor, M. (1993).
Bessel processes, Asian options and perpetuities.
\textit{Mathematical Finance}.
\textbf{3}, 349-375.

\bibitem{GobetMiri}
Gobet, E. and Miri, M. (2014). 
Weak approximation of averaged diffusion processes.
\textit{Stochastic Processes and their Applications}.
\textbf{124}, 475-504.

\bibitem{Hackmann2014}
Hackmann, D. and Kuznetsov, A. (2014).
Asian options and meromorphic L\'{e}vy processes.
\textit{Finance and Stochastics}.
\textbf{18}, 825-844.

\bibitem{HW} 
Henderson, V. and Wojakowski, R. (2002). 
On the equivalence of floating and fixed-strike Asian options, 
\textit{Journal of Applied Probability}. 
\textbf{39}, 391-394.

\bibitem{Kou}
Kou, S.G. (2002). 
A jump diffusion model for option pricing.
\textit{Management Science}.
\textbf{48}, 1086-1101.

\bibitem{KouWang}
Kou, S. G. and Wang, H. (2004).
Option pricing under a double exponential jump diffusion model.
\textit{Management Science}.
\textbf{50}, 1178-1192.


\bibitem{KyrQF}
Kyriakou, I., Pouliasis, P.K. and Papapostolou, N.C. (2016).
Jumps and stochastic volatility in crude oil prices and advances in average options pricing.
\textit{Quantitative Finance}.
\textbf{16}, 1859-1873.


\bibitem{Madan1990}
Madan, D.P. and Seneta, E. (1990). 
The VG for share market returns. 
\textit{Journal of Business}.
\textbf{63}, 511-524.

\bibitem{Madan1998}
Madan, D.P., Carr, P. and Chang, E. (1998).
The variance gamma process and option pricing.
\textit{European Finance Review}. 
\textbf{2}, 79-105.

\bibitem{Merton}
Merton, R.C. (1976).
Option pricing when the underlying stock returns are discontinuous.
\textit{Journal of Financial Economics}.
\textbf{3}, 125-144.

\bibitem{MN2011}
Muhle-Karbe, J. and Nutz, M. (2011).
Small-time asymptotics of option prices and first absolute moments.
\textit{Journal of Applied Probability}.
\textbf{48}, 1003-1020.

\bibitem{PagPasc}
Pagliarani, S. and Pascucci, A. (2013).
Local stochastic volatility with jumps: analytical approximations.
\textit{International Journal of Theoretical and Applied Finance}.
\textbf{16}(8), 1350050.

\bibitem{PZLV}
Pirjol, D. and Zhu, L. (2016).
Short maturity Asian options in local volatility models.
\textit{SIAM Journal on Financial Mathematics}.
\textbf{7}(1), 947-992.

\bibitem{PZIME}
Pirjol, D. and Zhu, L. (2016).
Discrete sums of geometric Brownian motions, annuities and Asian options.
\textit{Insurance: Mathematics and Economics}.
\textbf{70}, 17-36.

\bibitem{PZCEV}
Pirjol, D. and Zhu, L. (2019).
Short maturity Asian options for the CEV model.
\textit{Probability in the Engineering and Informational Sciences}.
\textbf{33}, 258-290.

\bibitem{AsianForwardStart}
Pirjol, D., J.~Wang and Zhu, L. (2019).
Short maturity forward start Asian options in local volatility models.
\textit{Applied Mathematical Finance}.
\textbf{26}(3), 187-221.

\bibitem{RogersShi}
Rogers, L. and Shi, Z. (1995).
The value of an Asian option.
\textit{Journal of Applied Probability}.
\textbf{32}, 1077-1088.


\bibitem{Varadhan67}
Varadhan, S. R. S. (1967).
Diffusion processes in a small time interval.
\textit{Communications on Pure and Applied Mathematics}.
\textbf{20}, 659-685.

\bibitem{VaradhanII} 
Varadhan, S. R. S. 
\textit{Large Deviations and Applications}, SIAM, Philadelphia, 1984.

\bibitem{Vecer}
Vecer, J. (2001).
A new PDE approach for pricing arithmetic average Asian options.
\textit{Journal of Computational Finance}. 
\textbf{4}, 105-113.

\bibitem{VecerXu}
Vecer, J. and Xu, M. (2002).
Unified Asian pricing.
\textit{Risk}. 
\textbf{15}, 113-116.

\end{thebibliography}
\end{document}